%% file: main.tex
\documentclass[sigconf]{acmart}

\input{commands}

\acmYear{2021}\copyrightyear{2021}
\setcopyright{rightsretained}
\acmConference[ICCPS '21]{ACM/IEEE 12th International Conference on Cyber-Physical Systems (with CPS-IoT Week 2021)}{May 19--21, 2021}{Nashville, TN, USA}
\acmBooktitle{ACM/IEEE 12th International Conference on Cyber-Physical Systems (with CPS-IoT Week 2021) (ICCPS '21), May 19--21, 2021, Nashville, TN, USA}
\acmPrice{}
\acmDOI{10.1145/3450267.3450534}
\acmISBN{978-1-4503-8353-0/21/05}
\begin{document}

\title{Probabilistic Conformance for Cyber-Physical~Systems}

\author{Yu Wang}
\email{yu.wang094@duke.edu}
\orcid{0000-0002-0431-1039}
\affiliation{%
  \institution{Duke University}
  \streetaddress{100 Science Dr.}
  \city{Durham}
  \state{NC}
  \postcode{27708}
}

\author{Mojtaba Zarei}
\email{mojtaba.zarei@duke.edu}
\orcid{}
\affiliation{%
  \institution{Duke University}
  \streetaddress{100 Science Dr.}
  \city{Durham}
  \state{NC}
  \postcode{27708}
}

\author{Borzoo Bonakdarpoor}
\email{borzoo@msu.edu}
\orcid{0000-0003-1800-5419}
\affiliation{%
  \institution{Michigan State University}
  \streetaddress{428 S. Shaw Ln.}
  \city{East Lansing}
  \state{MI}
  \postcode{48824}
}

\author{Miroslav Pajic}
\email{miroslav.pajic@duke.edu}
\orcid{0000-0002-5357-0117}
\affiliation{%
  \institution{Duke University}
  \streetaddress{100 Science Dr.}
  \city{Durham}
  \state{NC}
  \postcode{27708}
}


\begin{abstract}
%Conformance is an important concept in system analysis. It 
In system analysis, \emph{conformance} 
indicates that two systems simultaneously satisfy the same set of specifications of interest; thus, the results from analyzing one system automatically transfer to the other, or one system can safely replace the other in practice. In this work, we study the probabilistic conformance of cyber-physical systems (CPS). We propose a notion of (approximate) probabilistic conformance for sets of complex specifications expressed by the Signal Temporal Logic (STL). Based on a novel statistical test, we develop the first statistical verification methods for the probabilistic conformance of a wide class of CPS. Using this method, we verify the conformance of the startup time of the widely-used full and simplified model of Toyota powertrain systems, the settling time of model-predictive-control-based and neural-network-based automotive lane-keeping controllers, as well as the maximal voltage deviation of full and simplified power grid systems.
\end{abstract}

\maketitle

\input{intro}

\input{prelim}

\input{setup}

\input{ks}

\input{smc}

\input{eval}
\input{related}

\section{Conclusion}
\label{sec:conc}

In this paper, we proposed a new concept of {\em probabilistic conformance} for CPS. This notion is based on approximately equal satisfaction probabilities for a given (infinite) set of signal temporal logic (STL) formulas.
We introduced a verification algorithm for the probabilistic conformance of grey-box CPS, modeled by probabilistic uncertain systems. Our statistical verification algorithm is based on a new statistical test that can check if two probability distributions are equal for any desired confidence level (lower than~$1$). 
Finally, we used our approach to verify (1)~the nonconformity in the startup time of the full and simplified models of the Toyota powertrain system, (2)~the approximate conformity in the settling time of the model predictive control (MPC) based lane-keeping controller and neural network (NN)-based lane-keeping controllers of sufficient sizes, and
(3)~the nonconformity in the maximal DC voltage deviation of the full and simplified model of a power grid~system. An avenue for future work is to support conformance verification of systems for security/privacy policies that are {\em hyperproperties}. Besides, there is a need to go beyond verification and develop techniques to identify system behaviors that result in nonconformity.

\begin{acks}
This work is sponsored in part by the ONR under agreements N00014-17-1-2504 and N00014-20-1-2745, AFOSR under award number FA9550-19-1-0169, as well as the NSF CNS-1652544 and SaTC-1813388 awards.
\end{acks}

\bibliographystyle{ACM-Reference-Format}
\bibliography{ref}

\input{appendix}

\end{document}

%% file: commands.tex
\usepackage{amsthm}
\usepackage{amsmath}
\usepackage{url}
\usepackage{xurl}
\usepackage{xspace}
\usepackage{latexsym}
\usepackage{mathtools}
\usepackage{relsize}
\usepackage{resizegather}
\usepackage{pifont}
\usepackage{xcolor}
\usepackage{subcaption}
\usepackage{booktabs}
\usepackage[T1]{fontenc}
\usepackage[utf8]{inputenc}
\usepackage{algorithm}
\usepackage{algorithmicx}
\usepackage{algpseudocode}
\usepackage{gensymb}

\usepackage{todonotes}
\usepackage{listings}
\usepackage{hyperref}
\hypersetup{%
  colorlinks=true,           % use colored links
  allcolors=blue!70!black,   % use dark blue for all links
  pdfstartview=Fit,          % open PDF viewer with side-pane hidden.
  breaklinks=true,
}

\usepackage{tikz}
\usetikzlibrary{automata,positioning,fit}
\usetikzlibrary{shapes,arrows}
\tikzset{state/.style={circle, draw, minimum size=0.3cm, initial distance=0.2cm}}

\newcommand{\nat}{\mathbb{N}}

\newcommand{\real}{\mathbb{R}}

\newcommand{\nnreal}{\real_{\geq 0}}

\newcommand{\id}{\mathbf{I}}

\newcommand{\pr}{\mathbf{Pr}}

\newcommand{\abs}[1]{\vert #1 \vert}
\newcommand{\nm}[1]{\Vert #1 \Vert}

\renewcommand{\vec}[1]{\ensuremath{\underline{#1}}}

\newcommand{\assert}{\mathcal{A}}

\newcommand{\sys}{\mathcal{M}} % system
\newcommand{\sysStates}{\mathcal{X}} % system
\newcommand{\sysIns}{\mathcal{I}} % system inputs
\newcommand{\sysParas}{\mathcal{D}} % system inputs
\newcommand{\sysTran}{\mathcal{T}} % system inputs

\newcommand{\sysState}{X} % system
\newcommand{\sysInit}{X_\mathrm{init}} % system inputs
\newcommand{\sysIn}{I} % system inputs
\newcommand{\sysPara}{D} % system inputs
\newcommand{\sig}{\sigma}

\renewcommand{\G}{\Box}
\newcommand{\F}{\Diamond}

\renewcommand{\U}{\mathbin{\mathcal{U}}}

\newcommand{\True}{\ensuremath{\mathtt{T}}}
\newcommand{\False}{\ensuremath{\mathtt{F}}}

\usepackage{cleveref}
\newtheorem{assumption}{Assumption}
\newtheorem{definition}{Definition}
\newtheorem{theorem}{Theorem}
\newtheorem{remark}{Remark}

\newtheorem{example}{Example}
\newtheorem{lemma}{Lemma}

\crefname{section}{Section}{Sections}
\crefname{subsection}{Section}{Sections}
\crefname{definition}{Definition}{Definitions}
\crefname{proposition}{Proposition}{Propositions}
\crefname{lemma}{Lemma}{Lemmas}
\crefname{theorem}{Theorem}{Theorems}
\crefname{corollary}{Corollary}{Corollaries}
\crefname{example}{Example}{Examples}
\crefname{figure}{Figure}{Figures}
\crefname{condition}{Condition}{Conditions}
\crefname{remark}{Remark}{Remarks}
\crefname{running}{Running Example}{Running Examples}
\crefname{algorithm}{Algorithm}{Algorithms}
\crefname{table}{Table}{Tables}

%% file: intro.tex
\section{Introduction}
\label{sec:intro}

{\em Conformance} is an important concept in the analysis of cyber-physical systems (CPS)~\cite{ryabtsev2009translation,majumdar2013compositional,khakpour_NotionsConformanceTesting_2015,roehm_ReachsetConformanceTesting_2016,liu_ReachsetConformanceForward_2018,graf2019component}. It indicates that two systems satisfy the same set of given specifications (e.g., reachability or input-output relation). Thus the analysis results for one system can transfer to the other system, or one system can safely replace the other in practice. The term ``conformance'' may also refer to the consistency between a system and a design specification (e.g.,~\cite{heerink_FormalMethodsConformance_1996,lopez_SpecificationTestingImplementation_2006}); this is out of the scope of this~work.

For CPS, complex specifications for their dynamics are mathematically expressible by temporal logics, such as the Signal Temporal Logic (STL)~\cite{maler2004monitoring}. Following the line of work~\cite{abbas_FormalPropertyVerification_2014,deshmukh_QuantifyingConformanceUsing_2017}, we focus on the conformance of CPS for temporal logics specifications. This notion of conformance generalizes the conformance for reachability~\cite{roehm_ReachsetConformanceTesting_2016,liu_ReachsetConformanceForward_2018}, since reachability is expressible by temporal logic.

Conformance can be used for two different models derived from the same system under two conditions, implying that the system executes in the same way under the conditions (e.g., two inputs). A well-known example of nonconformity is the Volkswagen emissions scandal~\cite{barrett2015impact}, where the emission control software deliberately performs differently in the lab testing and driving conditions to bypass the emission test without actually reducing the pollution generated from the cars while driving. Similar undesirable nonconformity exists in printers~\cite{barthe_FacetsSoftwareDoping_2018}, where the software drivers deliberately work differently in favor of certain cartridge brands. To prevent such {\em software doping}~\cite{pinisetty_RuntimeVerificationHyperproperties_2018}, one needs to verify the conformance of a system under different conditions/settings.

The conformance also applies to two models derived from two systems operating under the same conditions, implying that they are interchangeable for the application. For instance, there has been recently significant interest in replacing precise but computationally expensive controllers based on model predictive control (MPC) with ones based on neural network (NN) for applications such as lane-keeping systems in autonomous vehicles~\cite{pan2017agile}. To migrate from an MPC controller to an NN controller without significantly changing the responsiveness, we need to check the conformance of the closed-loop system under the two controllers for their settling time, especially considering the fragility of AI-based controllers. While we focus on the conformance of two different systems operating under the same conditions in our case studies, our approach also applies to a system's conformance under two conditions.

Since CPS, such as autonomous vehicles, are frequently subject to randomness (e.g., system/network/environment noise), we propose a {\em probabilistic} notion of conformance for these systems. We use the definition of {\em probabilistic uncertain systems} (PUSs) from~\cite{wang_StatisticalVerificationHyperproperties_2019} to capture CPS dynamics. Roughly speaking, they are {\em grey-box} probabilistic dynamical systems with unknown dynamics in known state space. The PUSs capture the system nondeterminism as the input and probabilism as the parameters. The input and parameters can be time functions of general types, including real, integer, or categorical/Boolean. Given the input and parameters' value, a time-dependent sample path of general types can be generated. The PUSs subsume commonly used dynamical models such as continuous-time Markov chains and hybrid I/O automata~\cite{henzinger2000theory} with probabilistic parameters (used to capture the Toyota Powertrain~\cite{roohi_StatisticalVerificationToyota_2017}).

We define the notion of conformance through a parameterized signal temporal logic (STL) formula~\cite{asarin2011parametric} as illustrated in \cref{fig:chart1}. Specifically, we require that the satisfaction probabilities are approximately equal for all values of the STL parameters. For example, for the probabilistic conformance of two models $\sys_1$ and $\sys_2$ of reaching the same set $\mathcal{D}$, one can consider the parameterized formula $\F_{[0, t]} \mathcal{D}$ and require that for a given $c > 0$, it holds that
\[
\begin{split}
\forall t \in [0, \infty). \
\big| & \pr_{\sig_1 \sim \sys_1} (\sig_1 \models \F_{[0, t]} \mathcal{D} ) 
\\ & \quad - \pr_{\sig_2 \sim 
\sys_2} (\sig_2 \models \F_{[0, t]} \mathcal{D} ) \big| < c;        
\end{split}
\]
here, $\sig_1$ and $\sig_2$ are two random signals from the models $\sys_1$ and $\sys_2$, respectively, as illustrated in \cref{fig:hitting}. This implies that both systems $\sys_1$ and $\sys_2$ reach $\mathcal{D}$ with approximately equal probability for any time horizon. Our notion of conformance only requires these probabilities to be \emph{approximately} equal instead of \emph{exactly} equal, since the former is usually sufficient in practice (more examples are provided in \cref{sec:eval}).

\begin{figure}[t!]
    \centering
    \includegraphics[width=0.998\columnwidth]{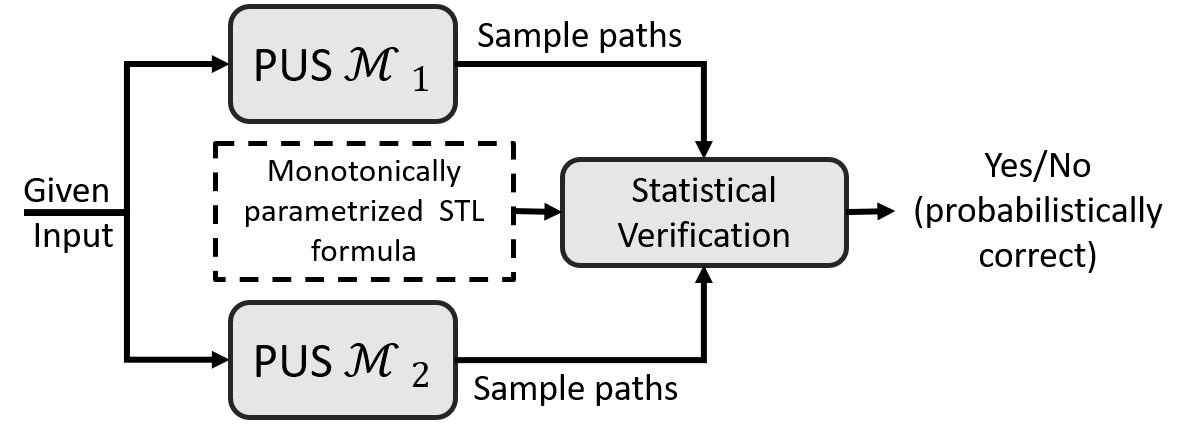}
    \caption{Overview of our statistical conformance~test.}
    \label{fig:chart1}
\end{figure}

Since the PUSs may have complex or even unknown dynamics, we adopt a statistical verification approach, as it scales better than model-based verification approaches and can handle unknown dynamics~\cite{agha_SurveyStatisticalModel_2018,larsen_StatisticalModelChecking_2016}.  From the conformance definition, we need to simultaneously handle the approximately equal satisfaction probability of infinitely many STL specifications since the parameters of the parameterized STL formula can take infinitely many values; this is very challenging since existing statistical verification methods can only handle a single (non-parametrized) temporal logic formula~\cite{legay_StatisticalModelChecking_2010,agha_SurveyStatisticalModel_2018} or a hyper temporal logic formula~\cite{wang2021csf,wang_StatisticalVerificationHyperproperties_2019}.

We show that statistically verifying conformance is feasible when the STL formula is \emph{monotonically} parameterized, i.e., the formula's satisfaction probability changes monotonically with the parameters. Such a property holds for many cases as discussed in detail in \cref{sec:formulation} and the case studies in \cref{sec:eval}. To the best of our knowledge, this work is the first to enable statistical verification for infinitely many formulas.

Due to monotonicity, the satisfaction probabilities over the values of its parameters on the two PUSs form two probability distributions. Accordingly, the conformance of two PUSs requires the two distributions to be approximately equal. To this end, we develop a new statistical test to check the \emph{approximate} equality of two distributions with provable confidence levels. Our test is based on the classic Kolmogorov-Smirnov (KS) test~\cite{deshpande_NonparametricStatisticsTheory_2018} and its multivariate generalization~\cite{peacock_TwodimensionalGoodnessoffitTesting_1983} for checking the \emph{exact} equality of two distributions. Based on this, we develop a statistical verification method for the probabilistic (non)conformance of two PUSs for any desired confidence level (lower than 1).

We apply the proposed statistical verification method to check the probabilistic conformance for three case studies to show its applicability. First, we study the probabilistic conformance of the widely used full and simplified models of the Toyota powertrain system~\cite{jin_PowertrainControlVerification_2014,roohi_StatisticalVerificationToyota_2017} for the startup time for their air to fuel ratio to reach a working region. Our results show the {\em nonconformity} of the two models, suggesting the simplified model may not capture certain important aspects of the system. Second, we check the probabilistic conformance of the settling time of an MPC-based lane-keeping controller and several NN-based lane-keeping controllers of different sizes for an autonomous car~\cite{MPCToolbox}. We show that NN-based controllers conform to the MPC-based controller, as their size increases; however, a small NN design may result in nonconformity. It suggests that an MPC-based controller can be replaced with a sufficiently-large NN-based controller to satisfyingly control the settling time.  Finally, we check the probabilistic conformance of the maximal deviation of DC voltage between the full model and a simplified model of a power grid system~\cite{PowerSim}. Our results show that the two models do not probabilistically conform -- i.e., the simplified model again may not capture certain important aspects of the system.

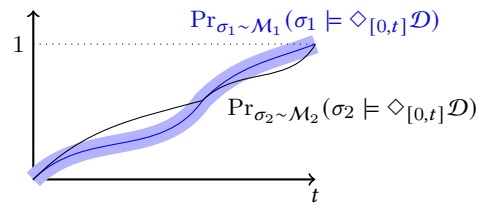
\begin{figure}[t!]
\centering
\begin{tikzpicture}[scale=1.50]
    \draw [->, thick] (0, 0) -- (0, 1.5);
    \draw [->, thick] (0, 0) -- (2.5, 0) node[below] () {$t$};
    \draw [line width=0.25cm, blue!30] (0, 0) to[out=45,in=-125] (1.5, 0.7) to[out=45,in=-160] (2.5, 1.2) node[above] () {};
    \draw [blue] (0, 0) to[out=45,in=-125] (1.5, 0.7) to[out=45,in=-160] (2.5, 1.2) node[above] () {$\pr_{\sig_1 \sim \sys_1} (\sig_1 \models \F_{[0, t]} \mathcal{D} )$};
    \draw [] (0, 0) to[out=45,in=-165] (1.5, 0.7) to[out=45,in=-120] (2.5, 1.2) node[below, yshift=-0.6cm, xshift=0.5cm] () {$\pr_{\sig_2 \sim \sys_2} (\sig_2 \models \F_{[0, t]} \mathcal{D})$};
    \draw [dotted] (0, 1.2) node[left] () {$1$} -- (2.5, 1.2);
    % \draw [dotted] (1.5, 0) node[below] () {$t_0$} -- (1.5, 0.7);
    \end{tikzpicture}
\caption{Reachability probabilities for some set $\mathcal{D}$ v.s. time horizon $t$. The two models conform (for reachability) if the black line stays within the blue tube.}
\label{fig:hitting}
\end{figure}

This paper is organized as follows. After preliminaries in \cref{sec:prelim}, in \cref{sec:formulation} we formalize the problem and our definition of probabilistic conformance for a parameterized STL formula. We present a new statistical test in \cref{sec:EKS}
and the verification method for the probabilistic conformance in \cref{sec:smc}. In \cref{sec:eval}, we apply our method to three real-world case-studies, before discussing related work in \cref{sec:related}, and concluding in \cref{sec:conc}.

\paragraph{Notation} We denote the sets of natural, real, and non-negative real numbers 
by $\nat$, $\real$, and $\nnreal$, respectively. We define $\real_\infty = \real \cup \{-\infty, \infty\}$, and $[n] = \{1,\ldots,n\}$, for $n \in \nat$. The cardinality and the power set of a set $S$ are denoted by $\abs{S}$ and $2^{S}$.

%% file: prelim.tex
\section{Problem Formulation} \label{sec:prelim}

We use a general system model for CPS called \emph{probabilistic uncertain systems} (PUSs)~\cite{wang_StatisticalVerificationHyperproperties_2019}. They capture continuous-time probabilistic dynamics on a hybrid state-space of discrete and continuous values, as well as generalize common probabilistic models such as continuous-time Markov chains (CTMC) and probabilistic hybrid I/O automata~\cite{wang_StatisticalVerificationHyperproperties_2019}. Since we adopt a statistical approach, we mainly view a PUS as a \emph{grey-box} that
generates random samples (\cref{fig:pus}).

\begin{definition} \label{def:PUS}
A probabilistic uncertain system (PUS) is a tuple $\sys = (\sysStates, \sysInit, \sysIns, \sysParas, \allowbreak \{\sysPara(t)\}_{t \in \nnreal}, \allowbreak \sysTran)$, where
\begin{itemize}
\item $\sysStates = \sysStates_1 \times \ldots \times \sysStates_n$ 
is the \emph{state space}
with each $\sysStates_i$ being either $\real$ or a discrete set 
$[n_{\sysStates_i}]$;

\item $\sysInit \in \sysStates$ is the \emph{initial state};

\item $\sysIns = \sysIns_1 \times \ldots \times \sysIns_m$ is the 
\emph{range of inputs} with each $\sysIns_i$ being either $\real$ or a discrete set $[n_{\sysIns_i}]$; 

\item $\sysParas = \sysParas_1 \times \ldots \times \sysParas_l$ is the 
\emph{range of parameters} with each $\sysParas_i$ being either $\real$ or a discrete set $[n_{\sysParas_i}]$;

\item $\{\sysPara(t)\}_{t \in \nnreal}$ is a random process on $\sysParas$ (for a properly defined probability space), defining the random change of the parameter over time;

\item $\sysTran: (\nnreal \to \sysIns) \times (\nnreal \to \sysParas) \to (\nnreal \to \sysStates)$ defines
the \emph{transition} of the system -- i.e., given the (time-dependent) value of the input and parameter,
the system deterministically generates a path.

\end{itemize}
\end{definition}

Given the value of the (time-dependent) \emph{input} 
$\sysIn: \nnreal \to \sysIns$, 
the PUS can generate a random \emph{signal}
$\sig(t) = \sysTran(\sysIn(t), \sysPara(t))$,
where the randomness comes from the parameter $\sysPara(t)$.
We denote by $\sig \sim \sys_\sysIn$ when the signal $\sig$ is randomly generated from the system $\sys$ for the given input $\sysIn$.
We also write $\sig \sim \sys$ if $\sysIn$ is clear
from~the~context.

There is no assumption on the dynamics of a PUS, such as Markovian, causal, etc. Common probabilistic models such as the discrete-time or continuous-time Markov chains~\cite{trivedi_ReliabilityAvailabilityEngineering_2017}, and probabilistic hybrid I/O automata~\cite{sproston_DecidableModelChecking_2000,zhang_SafetyVerificationProbabilistic_2010} are subsumed by the notion of PUS 
(see~\cite{wang_StatisticalVerificationHyperproperties_2019} for details).

\begin{figure}[!t]
\centering
\begin{tikzpicture}
	\draw (0, 0) ellipse (2 and 1);
	\node at (1.8, 0) {$\sysStates$};
	\draw [->, thick] (-1,-0.5) node[above, align=left] {$\sysState(t)$\\$t 
\in \nnreal$} to[out=0,in=-135] (0.7, 0.2);
	\draw [dotted] (-.5, .5) to (1.5, -.5);
	\draw [dotted] (.5, 0) to (.5, 1.3);
	\draw [dashed] (-.5, .55) to[out=-45,in=135] (0.6, 1) 
to[out=-45,in=165] 
(1.5, -.45);
	\node at (1.8, 0) {$\sysStates$};
	\node at (0, 1.4) {PUS $\sys$};
	\draw[->] (-3, 0) node[align=left] {Input\\$\sysIn(t) \in \sysIns$} to 
(-2.2, 0);
\end{tikzpicture}
\caption{Probabilistic Uncertain System (PUS).\label{fig:pus}}
\end{figure}
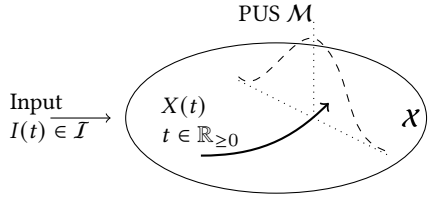

\begin{example}
A simple example of PUS is a bouncing ball
with random gravitational acceleration,
as shown in \cref{fig:bouncing}.
Its state is the height and velocity $(x, v)$.
For $x > 0$, the state evolves by
$\dot{x} = v, \dot{v} = g$;
for $x = 0$, it jumps by
$x \mapsto x, v \mapsto -v$.
The parameter $g$ is randomly drawn from
a normal distribution $N(g_0, \sigma^2)$ for some $g_0, \sigma > 0$.
The initial state is $(x_0, 0)$.
The input set is empty.
\end{example}

Finally, note that although by \cref{def:PUS}, a PUS has a unique initial state, it allows for defining conformance of paths from different initial states $\sysState_1$ and $\sysState_2$ of the PUS.
This is done by adding a new initial state $\sysState_0$ to the PUS, and model the transition from $\sysState_0$ to $\sysState_1$ and $\sysState_2$ as 
part of the input.

\paragraph*{Signal Temporal Logic}
We use the {\em signal temporal logic} (STL)~\cite{maler2004monitoring}
to capture the temporal specifications of interest 
for the random signals of the PUS.
STL can be viewed as the counterpart of linear temporal logic (LTL) in the real-time domain with real-valued constraints.
An STL formula is defined inductively by the syntax
\begin{equation} \label{eq:syntax}
	\varphi \Coloneqq \;
	f > 0 \; 
	\mid \neg \varphi \;
	\mid \varphi \land \varphi \;
	\mid \varphi \U_{[t_1, t_2]} \varphi, 
\end{equation}
where $f: \real^n \to \real$ is a given function.
To simplify further discussion, 
we let $t_1, t_2 \in \real_\infty$, 
instead of taking values in nonnegative rational numbers.
We call $f > 0$ an \emph{atomic proposition} and $\U_{[t_1, t_2]}$ the ``until'' operator.
Other temporal and logic operators are defined {as usual; for example,}
\begin{itemize}
	\item (false/true) $\False = \varphi \land (\neg \varphi)$ and $\True = \neg \False$,
% 	\item (disjunction) $\varphi \land \varphi = \neg (\neg \varphi \lor \neg \varphi)$, \todo{disjunction is defined but not other logic operators such as $\Rightarrow$; why not just say that the other logic operators are defined in a common way and remove disjunction}
	\item (finally) $\F_{[t_1,t_2]} \varphi = 	\True \ \U_{[t_1,t_2]} \varphi$, and
	\item (always) $\G_{[t_1,t_2]} \varphi = \neg (\F_{[t_1,t_2]} \neg\varphi)$.
\end{itemize}

For a concrete \emph{signal} $\sig: \nnreal \to \real^n$ of the PUS,
the satisfaction relation for STL formulas is defined recursively by
the semantics
\[
\begin{array}{l@{\hspace{0.5em}}c@{\hspace{0.5em}}l}
\sig \models f > 0 & \textrm{ iff } & f(\sig(0)) > 0\\ 
\sig \models \neg \varphi & \textrm{ iff } & \sig \not\models \varphi \\ 
\sig \models \varphi_1 \land \varphi_2 & \textrm{ iff } & \sig \models \varphi_1 
\textrm{ and } \sig \models \varphi_2 \\
 \sig \models \varphi_1  \U_{[t_1, t_2]} \varphi_2 & \textrm{ iff } &
\textrm{there exists } t \in [t_1, t_2] \textrm{ such that }\\
& & \hspace{0.2cm} \sig^{(t)} \models 
\varphi_2  \textrm{ and for any } 0 \leq t' < t, \\
& & \hspace{0.2cm} \textrm{ it holds that } \ 
\sig^{(t')} \models \varphi_1;
\end{array}
\]
here, $\sig^{(t)}$ denotes the $t$-shift of the signal, defined by $\sig^{(t)}(t') = 
\sig(t + t')$ for any $t' \in \nnreal$.

We make the \textbf{convention} that a formula $\varphi_1  \U_{[t_1, t_2]} \varphi_2$ 
is equivalent to $\False$, if $t_2 < t_1$, $t_1 < 0$, or $t_2 < 0$. 

\begin{example}
The following STL formula requires that if $\abs{x} > 0.5$, then within $0.6$ time units $\abs{x}$ settles under the value $0.5$ for the $1.5$-long time~interval 
$$\varphi = \G\Big(\abs{x}  > 0.5 \Rightarrow \F_{[0,.6]}(\G_{[0,1.5]} 
\abs{x} < 0.5) \Big).$$
\end{example}

\begin{figure}[t!]
\centering
\includegraphics[width=0.4\columnwidth]{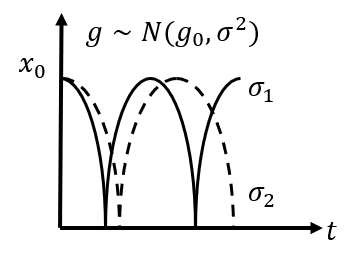}
\caption{Stochastic bouncing ball.}
\label{fig:bouncing}
\end{figure}

%% file: setup.tex
\section{Probabilistic Conformance} % and Problem Formulation}
\label{sec:formulation}

% Following the line of work from~\cite{abbas_FormalPropertyVerification_2014,deshmukh_QuantifyingConformanceUsing_2017}, 
We focus on a class of conformance properties for CPS for \emph{an (infinite) set of STL formulas}.
Mathematically, we say that two PUSs probabilistically conform if for \textbf{any} STL formula from the set, the satisfaction probabilities are \emph{approximately equal} for two random signals drawn respectively from the two PUSs.
This can be viewed as a probabilistic generalization of~\cite{abbas_FormalPropertyVerification_2014,
deshmukh_QuantifyingConformanceUsing_2017}.

% \begin{figure}[!t]
% 	\centering
% 	\begin{tikzpicture}
% 		\node[draw, rectangle, align=left] at (0, 0) (s1) {PUS $\sys$};
% 		\node[draw, rectangle, align=left] at (0, 1) (s2) {PUS $\sys$};
% 		\node[draw, dashed, rectangle, align=center] at (3.5, 0.5) (smc) {$\forall \varphi \in \Phi, \ \pr(\vec{\sysState_1} \models \varphi)$\\$\approx_d \pr(\vec{\sysState_2} \models \varphi)$};
		
% 		\draw[->] (-1.3, 0) node[align=left] {Input\\$\sysIn_1 \in \sysIns$} to (s1);
% 		\draw[->] (-1.3, 1) node[align=left] {Input\\$\sysIn_2 \in \sysIns$} to (s2);
% 		\draw[->] (s2) to node[above] {$\vec{\sysState_1}$} (smc);
% 		\draw[->] (s1) to node[below] {$\vec{\sysState_2}$} (smc);

% 	\end{tikzpicture}
% 	\caption{Approximate probabilistic conformance of PUS under a finite set of relevant inputs $\sysIn$ for a monotone formula set $\Phi$ of temporal specifications.\label{fig:conformance}}
% \end{figure}

\begin{definition}[Conformance] \label{def:conformance}
Let $\Phi$ be an infinite set of STL formulas. 
For two PUSs $\sys_1$ and $\sys_2$, % from \cref{def:PUS} 
and a given $c > 0$,
we say that $\sys_1$ and $\sys_2$ $c$-approximately probabilistically conform for $\Phi$
(for the same given input),
if for any STL formula $\phi \in \Phi$, it holds that
\[
	\big| \pr_{\sig_1 \sim \sys_1} (\sig_1 \models \phi) - 
	\pr_{\sig_2 \sim \sys_2}(\sig_2 \models \phi) \big| < c,
\]
where $\sig_i \sim \sys_i$ is a random path from the PUS $\sys_i$, for $i \in \{1, 2\}$.  
\end{definition}

In \cref{def:conformance}, we only require the satisfaction probabilities to be \emph{approximately} equal for the STL formulas of interest instead of \emph{exactly} equal; the latter is usually unnecessary in applications (see e.g. the case studies presented in \cref{sec:eval}).
Besides, the conformance from \cref{def:conformance} cannot be expressed by single formulas in any common temporal logic since a parameterized formula effectively captures an infinite number of STL formulas. 
For any fixed values of the employed parameters, the property can be expressed in HyperPSTL~\cite{wang_StatisticalVerificationHyperproperties_2019}.

Depending on the choice of the class (i.e., set) of temporal properties $\Phi$, different 
notions for the conformance of PUS are derived, including probabilistic reach-set 
conformance and probabilistic trace conformance.
Commonly, an STL formula set $\Phi$ 
can be derived 
by parametrizing a single STL formula $\phi$ by~\cite{asarin2011parametric} %\todo{fix the reference} 
\begin{equation} \label{eq:Phi}
	\Phi = \{\phi_{\vec{d}}: \vec{d} \in \real^K\}.
\end{equation}
Effectively, $\phi_{\vec{d}}$ represents infinitely many STL formulas, as the parameter $\vec{d}$ can take infinitely many values.

For example, the STL formula set 
\begin{equation} \label{eq:stlset1}
	\Phi_1 = \{ \F_{[0, 1]} (\sig > a): a \in \real \}
\end{equation}
is derived by parametrizing the threshold $a$.
It contains an infinite set of reachability specifications
for the parametrized threshold~$a$ within the fixed time-interval $[0,1]$.
The conformance of the two PUSs $\sys_1$ and $\sys_2$ 
for the set $\Phi_1$ means that, for any threshold $a$  
the probability of reaching 
the threshold should be approximately equal for two random signals respectively from $\sys_1$ and $\sys_2$.

Similarly, the STL formula~set 
\begin{equation} \label{eq:stlset2}
	\Phi_2 = \{ \F_{[0, t]} (\sig > 0): t \in \real \}
\end{equation}
is derived by parametrizing the time horizon $t$.
It contains an infinite set of reachability specifications
for the fixed threshold $0$, within a parameterized time interval $[0,t]$.
The conformance of the two PUSs $\sys_1$ and $\sys_2$ 
for the set $\Phi_2$ means that
the probability of reaching the threshold 0 (i.e., $> 0$)
within any time interval $[0,t]$
should be approximately equal for two random signals respectively from $\sys_1$ and $\sys_2$.

Considering that the PUSs can have complex dynamics 
that may be even unknown in practice,
in this work we propose to statistically verify the conformance of PUSs from \cref{def:conformance}; such method exhibits better scalability than the exhaustive 
approaches and can handle unknown dynamics
\cite{agha_SurveyStatisticalModel_2018,
larsen_StatisticalModelChecking_2016}.
There are infinitely many STL formulas 
of interest in \eqref{eq:Phi}, so the proposed statistical verification method
should be able to handle an infinite set of STL specifications.
This is very challenging since all existing statistical verification
techniques can only handle single STL specifications
\cite{legay_StatisticalModelChecking_2010,agha_SurveyStatisticalModel_2018}.
To solve this, we focus on the conformance for \emph{monotonically} parameterized STL formulas, {which are commonly used in system~analysis}~\cite{asarin2011parametric}.

% \paragraph*{Monotonically parameterized STL formulas}
% \yw{Monotonically parameterized STL formulas are an important class of parameterized STL formulas~\cite{asarin2011parametric}}. 
%\todo{I removed the "paragraph*".}
%
Generally, the parametrized formula $\phi_{\vec{d}}$ 
(where $\vec{d}$ captures the vector of parameters)
is monotone if the satisfaction probability on a model 
is preserved for the order of the parameters -- 
i.e., the satisfaction probability changes monotonically with the parameter. 
While statistically verifying the probabilistic conformance for
an arbitrary STL formula set is very difficult,
handling a monotonically parameterized formula set
can be done by exploiting the formula's~monotonicity.

\begin{definition}[Monotonically Parameterized Formula] \label{def:mfs}
A parameterized formula $\phi_{\vec{d}}$
with $\vec{d} \in \real^K$
is {\em monotone} for a PUS $\sys$ if for any given path $\sig$ from $\sys$ and $i \in [K]$, and
\begin{itemize}
	\item 
	for any $\vec{d}, \vec{d'}$ such that $\vec{d} \preceq_i \vec{d'}$,
	it holds that
	$\sig \models \phi_{\vec{d}}$ implies $\sig \models \phi_{\vec{d'}}$, OR 

	\item 
	for any $\vec{d}, \vec{d'}$ such that $\vec{d} \preceq_i \vec{d'}$,
	it holds that $\sig \models \phi_{\vec{d'}}$ implies $\sig \models \phi_{\vec{d}}$;	
\end{itemize}	
{here, $\vec{d} \preceq_i \vec{d'}$ denotes that
the entries of $\vec{d}$ and $\vec{d'}$ are equal 
except for $\vec{d}_i \leq \vec{d'}_i$.}
\end{definition}
	
Following \cref{def:mfs}, the parameter alternation preserves the parametrized STL formula's monotonicity.

\begin{definition}[Alternation] \label{def:alternation}
The function $\pi(\vec{d}) = \vec{d}'$ is called an alternation, if for all $i \in [K]$, $d'_i = d_i$ or $d'_i = - d_i$. The set of all $K$-dimensional alternations in $\real^K$ is denoted by $\Pi_K$.
\end{definition}

From the previous definitions, the following directly holds. 

\begin{lemma}
If $\phi_{\vec{d}}$ is a monotonically parameterized STL formula,
then so is $\phi_{\pi({\vec{d})}}$, where $\pi$ is an alternation.%
\footnote{The alternation of time horizon parameters in $\U$ (and other temporal operators) follows the aforementioned convention that $\varphi_1  \U_{[t_1, t_2]} \varphi_2$ is equivalent to $\False$, if $t_2 < t_1$, $t_1 < 0$, or $t_2 < 0$. }
\end{lemma}

% The monotonicity of a parameterized STL formula may depend on the model $\sys$.
% %
% For example, the parameterized STL formula from~\eqref{eq:stlset2}
% is monotone for $t \in \real$ for any PUS, since from STL semantics,
% $\F_{[0, t_1]} (\sig > 0)$ 
% always implies
% $\F_{[0, t_2]} (\sig > 0)$  
% for any $t_2 \geq t_1$ (including negative $t_1, t_2$).
% %
% On the other hand, the parameterized STL formula from~\eqref{eq:stlset1} is monotone for $a \in \real$, if any signal $\sig(t)$ from the PUS $\sys$ is monotone -- i.e., for any $t_2 \geq t_1$, it holds that $\sig(t_2) \geq \sig(t_1)$ or $\sig(t_2) \leq \sig(t_1)$.
% %
% As illustrated in \cref{fig:monotone signal},
% if signal $x(t)$ is non-decreasing
% then for any $a_1 \geq a_2$,
% it~holds
% \begin{equation} \label{eq:ex2}
% 	\big(\F_{[0,1]} \G (\sig > a_1)\big) \Rightarrow \big(\F_{[0,1]} \G (\sig > a_2)\big).
% \end{equation}
% %
% This is generally not true for an arbitrary signal $\sig(t)$.

%% file: ks.tex
\section{Statistical Test for Approximate Equality of Distributions} %Probability~
\label{sec:EKS}

Before introducing a statistical verification algorithm for probabilistic conformance, we propose a new statistical test for the equivalence of two (unknown) probability distributions, based on the classic Kolmogorov-Smirnov~test~\cite{deshpande_NonparametricStatisticsTheory_2018,peacock_TwodimensionalGoodnessoffitTesting_1983}.
We start from the scalar case and then extend to the multidimensional case.
%\todo{rephrased the commented paragraph below.}

% We now introduce a statistical verification algorithm for verifying the probabilistic conformance of two PUSs for a monotonically parametrized STL formula (i.e., for the corresponding STL formula set). Due to the monotonicity, the satisfaction probabilities on a PUS of the parametrized formula can be captured by an (unknown) cumulative distribution function (CDF). To check the probabilistic conformance for this monotonically parametrized formula, it suffices to check the equivalence of two (unknown) CDFs. To achieve this, we propose a new statistical test, based on the classic Kolmogorov-Smirnov~test~\cite{deshpande_NonparametricStatisticsTheory_2018}.

% \begin{figure}[!t]
% \centering
% \begin{tikzpicture}[scale=1.50]
% \draw [->, thick] (0, 0) -- (0, 1.5) node[below left] () {$x(t)$};
% \draw [->, thick] (0, 0) -- (3, 0) node[below] () {$t$};
% \draw [] (0, 0) to[out=45,in=-125] (1.2,0.7) to[out=45,in=-160] (2.5, 1.2) node[above] () {};
% \draw [dashed] (0, 0.6) -- (3, 0.6) node[above] () {$a_2$};
% \draw [dashed] (0, 1) -- (3, 1) node[above] () {$a_1$};
% \end{tikzpicture}
% \caption{Example Monotone Path.
% \label{fig:monotone signal}}
% \end{figure}

% \subsection{Statistical Test for Conformance} 

Consider two $K$-dimensional random vectors 
$\vec{X} = (X_1, \allowbreak \ldots, X_K)$ 
and $ \vec{Y} = (Y_1, \ldots, Y_K)$.
For each $K$-dimensional alternation $\pi \in \Pi_K$, we~define 
\begin{equation} \label{eq:CDF}
\begin{split}
    & F^{\pi} (\vec{a}) = \pr_{ \vec{X}} \big(\pi(\vec X)_1 \leq \pi(\vec{a})_1, \ldots, \pi(\vec X)_K \leq \pi(\vec{a})_K\big),
 \\ & G^{\pi} (\vec{a}) = \pr_{ \vec{Y}} \big(\pi(\vec Y)_1 \leq \pi(\vec{a})_1, \ldots, \pi(\vec Y)_K \leq \pi(\vec{a})_K\big),
\end{split}
\end{equation}
where $\pi(\vec X)_i$ is the $i^{th}$ entry of $\pi(\vec X)$,
and the probabilities $\pr_{\vec{X}}$
and $\pr_{ \vec{Y}}$
are taken for the random vectors 
$\vec{X}$ and $\vec{Y}$, respectively.
If $\pi$ is the identity map, 
then $F^{\pi}$ and $G^{\pi}$ are respectively
the cumulative distribution functions (CDFs)
of $\vec{X}$ and $\vec{Y}$,
which we denote by $F$ and $G$ to simplify our notation.
Otherwise, $F^{\pi}$ and $G^{\pi}$ are 
the complimentary CDFs of $\vec{X}$ and $\vec{Y}$.

To measure the \emph{difference} between the probability distributions of $\vec{X}$ and $\vec{Y}$, let 
\begin{equation} \label{eq:alpha_XY}
    \gamma_{\vec{X},\vec{Y}} = \max_{\pi \in \Pi_K} \nm{F^{\pi} - G^{\pi}}_\infty,
\end{equation}
with $\nm{\cdot}_\infty$ standing for the $L_\infty$ function norm. 
If $\gamma_{\vec{X},\vec{Y}} = 0$, then 
$\vec{X}$ and $\vec{Y}$ have the same probability distributions.
 
The approximate equality of the probability distributions 
of $\vec{X}$ and $\vec{Y}$ is formulated as 
the hypothesis testing problem %\todo{no $=$?}
\begin{equation} \label{eq:HT}
    \mathcal{H}_0: \gamma_{\vec{X},\vec{Y}} < c \quad \mathcal{H}_1: \gamma_{\vec{X},\vec{Y}} > c, 
\end{equation}
where $c > 0$ is the given parameter for approximate equality. 
The alternation~$\pi$ in~\eqref{eq:alpha_XY} is necessary since two different multidimensional probability distributions may have the same CDFs but different complimentary CDFs.

\begin{assumption}
Similar to previous work on statistical verification~\cite{wang_StatisticalVerificationHyperproperties_2019,zarei2020statistical}, we assume $\gamma_{\vec{X},\vec{Y}} \neq c$, which ensures that as the number of samples increases, the samples will increasingly concentrate to support either $\mathcal{H}_0$ or $\mathcal{H}_1$ by the central limit theorem. Therefore, a statistical analysis based on the majority of the samples has increasing accuracy. This assumption is weaker than the ``indifference region'' adopted in other works on statistical verification~\cite{legay_StatisticalModelChecking_2010,agha_SurveyStatisticalModel_2018}.
\end{assumption}

\begin{remark}
The hypothesis testing problem~\eqref{eq:HT}
cannot be handled by 
the classic Kolmogorov-Smirnov (KS) test~\cite{deshpande_NonparametricStatisticsTheory_2018} 
and its multivariate generalization~\cite{peacock_TwodimensionalGoodnessoffitTesting_1983},
since they can only check for \emph{the exact equality} of 
two probability distributions, i.e., the hypothesis testing problem
% \todo{why did you prune the previous discussion?}
\begin{equation} \label{eq:HT_classic}
    \mathcal{H}'_0: \gamma_{\vec{X},\vec{Y}} = 0 \quad \mathcal{H}'_1: \gamma_{\vec{X},\vec{Y}} > 0. 
\end{equation}
\end{remark}

To solve problem \eqref{eq:HT}, we build on the KS test and introduce a new statistical test for any given confidence level $\alpha$ (i.e., the lowest probability that the test's assertion agrees with the truth in all cases).%
\footnote{The confidence level $\alpha$ is minimum of the p-values of the two hypothesis. Accordingly, we refer to $1 - \alpha$ as the significance level, which is the maximal of the false positive and the false negative rates.}
To facilitate presentation, we start from the scalar case and then move to the vector case. 

% two cases, where $\vec{X}$ and $\vec{Y}$ are scalar or vector, starting from the former.

\subsection{Scalar Random Variables}

If $X$ and $Y$ are scalar,%
\footnote{For this scalar case, to simplify our notation, we denote $\vec X$ and $\vec Y$ as $X$ and $Y$.}
then from~\eqref{eq:HT},
we have that $\gamma_{X,Y} = \nm{F - G}_\infty$,
where $F$ and $G$ are the CDFs of $X$ and $Y$, respectively.% 
\footnote{This does not hold in general for multidimensional random variables.}
Given two sets of independent and identically distributed (i.i.d.) samples 
$$X^{[n]} = ( {X}^{(1)}, \ldots, {X}^{(n)} ), \quad Y^{[m]} = ( {Y}^{(1)}, \ldots, {Y}^{(m)} ),$$ 
drawn respectively from $X$ and $Y$, 
%the empirical cumulative distribution functions (ECDFs) 
the ECDFs of the samples are
\begin{equation} \label{eq:ecdf}
\begin{split}
    & F_{X^{[n]}} (x) = \frac{1}{n} \sum\nolimits_{i = 1}^n \id (X^{(i)} \leq x), 
    \\ & G_{Y^{[m]}} (y) = \frac{1}{m} \sum\nolimits_{i = 1}^m \id (Y^{(i)} \leq y),
\end{split}
\end{equation}
where $\id (\cdot)$ is the indicator function.
Intuitively, the different $\gamma_{X,Y}$ can be statistically estimated by (as illustrated in \cref{fig:ks test})
\begin{equation} \label{eq:deltanm}
    \delta_{X^{[n]}, Y^{[m]}} = \nm{F_{X^{[n]}}  - G_{Y^{[m]}}}_\infty. 
\end{equation}
When the numbers of samples $n, m \to \infty$,
the ECDFs converges to 
the CDFs: $F_{X^{[n]}}  \to F$ and $G_{Y^{[m]}}  \to G$,%
\footnote{More precisely, this is convergence in distribution.}
and thus,
$\delta_{X^{[n]}, Y^{[m]}} \to \gamma_{X,Y}$ by Glivenko-Cantelli theorem~\cite{van1996glivenko}.
Therefore, for the hypothesis testing problem \eqref{eq:HT},
we propose the statistics assertion
\begin{equation} \label{eq:assert}
    \assert(X^{[n]}, Y^{[m]}) = \begin{cases}
        \mathcal{H}_0, & \text{if } \delta_{X^{[n]}, Y^{[m]}} < c, \\
        \mathcal{H}_1, & \text{if } \delta_{X^{[n]}, Y^{[m]}} > c.
    \end{cases}
\end{equation}

For random samples $X^{[n]}$ and $Y^{[m]}$, 
the probability $\alpha$ that the assertion \eqref{eq:assert} agrees with 
the correct answer to the hypothesis testing problem \eqref{eq:HT}
is called the \emph{confidence level}.
It depends on the \emph{discrepancy} between $\gamma_{X,Y}$ and $\delta_{X^{[n]}, Y^{[m]}}$,
which is bounded by 
\begin{equation} \label{eq:delta} 
    d_{X^{[n]}, Y^{[m]}} = \nm{(F_{X^{[n]}}  - F) - (G_{Y^{[m]}}  - G)}_\infty
\end{equation}
due to the triangle inequality
\begin{equation} \label{eq:triangle}
    \big| \delta_{X^{[n]}, Y^{[m]}} - \gamma_{X,Y} \big| \leq d_{X^{[n]}, Y^{[m]}}.
\end{equation}
When the numbers of samples $n, m \to \infty$, the discrepancy $d_{X^{[n]}, Y^{[m]}} \allowbreak \to 0$ with probability $1$.
However, the probability distribution 
of the rescaled discrepancy 
$d_{X^{[n]}, Y^{[m]}} \allowbreak \sqrt{mn/(m+n)}$
(for random samples $X^{[n]}, Y^{[m]}$)
is asymptotically invariant of $n, m$ 
and is independent of the CDFs $F$ and $G$, 
as formally stated below.

\begin{lemma}[Section {7.9} of 
\cite{deshpande_NonparametricStatisticsTheory_2018}] \label{lem:independence}
The CDF $H (x)$ of the $d_{X^{[n]}, Y^{[m]}} \allowbreak \sqrt{mn/(m+n)}$ from~\eqref{eq:delta}
obeys the Kolmogorov-Smirnov distribution
\begin{equation} \label{eq:delta_cdf}
    H (x) = 
    1 - 2 \sum_{i=1}^\infty (-1)^{i-1} e^{-2 i^2 x^2}
    \approx 1 - 2 e^{-2 x^2}.
\end{equation}
\end{lemma}

\begin{figure}[!t]
\centering
\begin{tikzpicture}[scale=1.50]
\draw [->] (0, 0) -- (0, 1.8);
\draw [->] (-2, 0) -- (2, 0) {};
\draw [dashed] (-2, 1.5) -- (0.1, 1.5) node[above] () {1} -- (2, 1.5);
\draw [thick] (-2, 0) to (-1.5, 0) to (-1.5, 0.2) to (-0.6, 0.2) to (-0.6, 0.5) to (0, 0.5) to (0, 1) to node[below] () {$F_{X^{[n]}} (x)$} (1, 1) to (1, 1.5) to (2, 1.5);
\draw [dash dot, thick] (-2, 0) to (-1.5, 0) to (-1.5, 0.3) to (-1.1, 0.3) to (-1.1, 0.8) to node[above, xshift=-0.2cm] () {$G_{Y^{[m]}} (x)$} (-0.5, 0.8) to (-0.5, 1.2) to (0, 1.2) to (1.5, 1.2) to (1.5, 1.5) to (2, 1.5);
\draw [<->, dotted] (-0.3, 0.5) -- node[yshift=-0.3cm, xshift=-0.2cm] () {\small $\delta_{X^{[n]}, Y^{[m]}}$} (-0.3, 1.2);
\end{tikzpicture}
\caption{Illustration of the statistics $\delta_{X^{[n]}, Y^{[m]}}$.\label{fig:ks test}}
\end{figure}
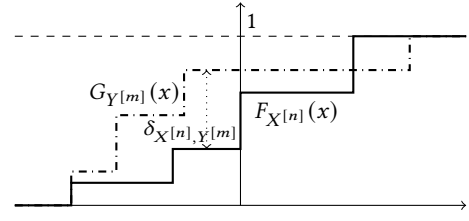

Now, we derive the significance level (i.e., $1 -\alpha$, where $\alpha$ is the confidence level ) of the assertion~\eqref{eq:assert} when observing a value of the test statistics $\delta_{X^{[n]}, Y^{[m]}}$ to be $\lambda$. If $\lambda < c$, the significance level of~\eqref{eq:assert} is the probability of observing a value of $\delta_{X^{[n]}, Y^{[m]}}$ (in any other test) that is lower than $\lambda$ (i.e., more in favor of the hypothesis $\mathcal{H}_0$) under the likely-false hypothesis $\mathcal{H}_1$. It holds that%
\footnote{{To simplify our notation, we also use $\delta_{X^{[n]}, Y^{[m]}}$ to denote a random value of the test statistics in any other test in computing the significance level.}}
\begin{align}
1 - \alpha & = \pr_{X^{[n]}, Y^{[m]}} \big( \delta_{X^{[n]}, Y^{[m]}} < \lambda \mid \mathcal{H}_1 \big) \notag
\\ & = \pr_{X^{[n]}, Y^{[m]}} \big( \gamma_{X,Y} - \delta_{X^{[n]}, Y^{[m]}} > \gamma_{X,Y} - \lambda \mid \mathcal{H}_1 \big) \notag
\\ & \leq \pr_{X^{[n]}, Y^{[m]}} \big( d_{X^{[n]}, Y^{[m]}} > \gamma_{X,Y} - \lambda \mid \mathcal{H}_1 \big) \label{aux:1}
\\ & \leq \pr_{X^{[n]}, Y^{[m]}} \big( d_{X^{[n]}, Y^{[m]}} > c - \lambda \big) \label{aux:2}
\\ & = 1 - H \big( (c - \lambda) \sqrt{mn/(m+n)} \big) \label{aux:3}
\end{align}
where~\eqref{aux:1} follows from~\eqref{eq:triangle};
\eqref{aux:2} holds since $\gamma_{X,Y} > c$ under $\mathcal{H}_1$;
and~\eqref{aux:3} follows from~\cref{lem:independence}.
Similarly, if $\lambda > c$, the significance level of~\eqref{eq:assert} satisfies
\begin{equation}
1 - \alpha \leq 1 - H \big( (\lambda - c) \sqrt{mn/(m+n)} \big). \label{aux:4}
\end{equation}
Finally, combining~\eqref{aux:3} and~\eqref{aux:4}, the confidence level of~\eqref{eq:assert} satisfies
\begin{equation} \label{eq:alpha}
    \alpha \geq  H \big( \vert \lambda - c \vert \sqrt{mn/(m+n)} \big).
\end{equation}

Based on~\eqref{eq:alpha}, for any desired confidence level $\alpha_d < 1$, our statistical test is deployed \emph{sequentially}. It can return an assertion with an actual confidence level of at least $\alpha_d$. Iteratively, the algorithm draws $k_1$ and $k_2$ new samples from the two CDFs $F$ and $G$, respectively, and then computes the actual confidence level $\alpha$ from~\eqref{eq:alpha}. It terminates when $\alpha > \alpha_d$, and then returns the assertion by \eqref{eq:assert}. This is formally captured in \cref{alg:EKS}.

\begin{algorithm}[!t]
    \caption{Proposed statistical test.\label{alg:EKS}}
    \begin{algorithmic}[1]
    \Require Desired confidence $\alpha_d > 0$, $c \in (0,1)$, $k_1, k_2 \in \nat$.
    
    \State Sample sizes $n, m \gets 0$, $\alpha \gets 0$.
    
    \While{$\alpha < \alpha_d$}
        
    \State Draw $k_1$, $k_2$ new samples from $X$, $Y$, respectively.

    \State $n \gets n + k_1$, $m \gets m + k_2$
    
    \State Update $\delta_{X^{[n]}, Y^{[m]}}$ by \eqref{eq:deltanm}.
    
    \State Update $\alpha$ by \eqref{eq:alpha}.
    
    \EndWhile
    
    \State \Return $\assert$ by \eqref{eq:assert}.
    
    \end{algorithmic}
\end{algorithm}

\begin{theorem} \label{thm:1}
\cref{alg:EKS} terminates with probability $1$ and has the confidence
level $\alpha_d$.
\end{theorem}

\begin{proof}
{\bf Termination}:
As $n, m \to \infty$, 
we have $\delta_{X^{[n]}, Y^{[m]}} \to \gamma_{X,Y} \neq c$,
{ -- i.e., $\delta_{X^{[n]}, Y^{[m]}}$ converges to some value that is not $c$ with 
probability $1$},
so either $\mathcal{H}_0$ or $\mathcal{H}_1$ holds.
Therefore, \cref{alg:EKS} terminates with probability $1$.

\noindent {\bf Correctness}:
Let $\tau$ be the step \cref{alg:EKS} terminates
and $A$ be ``the assertion $\assert$ from~\eqref{eq:assert} is correct'',
then
$\pr(A) = \sum\nolimits_{i \in \nat} \pr(A   \vert \, \tau = i) \pr(\tau = i)$.
From~\eqref{eq:alpha}, for any $i \in \nat$, 
we have that $\pr(A \allowbreak \, \vert \, \tau = i) > \alpha_d$.
In addition, by {\bf Termination}, we have that $\sum_{i \in \nat} \pr(\tau = i) = 1$, 
Therefore, it holds that $\Pr(A) \geq \alpha_d$.
\end{proof}

\begin{remark}
Although the test statistics $\delta_{X^{[n]}, Y^{[m]}}$ from~\eqref{eq:deltanm} is also used in the standard KS test~\cite{deshpande_NonparametricStatisticsTheory_2018}, the implementation and thresholding on $\delta_{X^{[n]}, Y^{[m]}}$ in~\eqref{eq:assert} in our method fundamentally differs from the KS test. Specifically, our statistical test increasingly draws samples until reaching the desired confidence level (< 1) and the thresholding on $\delta_{X^{[n]}, Y^{[m]}}$ in~\eqref{eq:assert} represents the similarity of the probability distributions as given in~\eqref{eq:HT}. On the other hand, the KS test employs a fixed number of samples and the thresholding on $\delta_{X^{[n]}, Y^{[m]}}$ is related to the confidence level. Consequently, our method guarantees confidence levels for both $\mathcal{H}_0$ and $\mathcal{H}_1$ in~\eqref{eq:HT}, while the KS test only guarantees confidence level for $\mathcal{H}'_0$, and not $\mathcal{H}'_1$, in~\eqref{eq:HT_classic}.
\end{remark}

\begin{algorithm}[!t]
\caption{Statistical verification for conformance.\label{alg:smc}}
\begin{algorithmic}[1]
\Require Desired confidence level $\alpha_d$, threshold $c > 0$ 

\State Sample sizes $n, m \gets 0$, $\alpha \gets 0$ 
\While{$\alpha < \alpha_d$}
    
\State Draw new samples from $\sys_1, \sys_2$ and update $n, m$.

\State Update $F^{\pi}_{n}, G^{\pi}_{m}$ by~\eqref{eq:Fn} and compute $\delta_{n,m}$ by~\eqref{eq:mD delta}.  

\State Update $\alpha$ by~\cite{peacock_TwodimensionalGoodnessoffitTesting_1983,fasano_MultidimensionalVersionKolmogorovSmirnov_1987}.

\EndWhile
\State \Return $\assert$ by~\eqref{eq:assert}.
\end{algorithmic}
\end{algorithm}

\subsection{Multidimensional Random Variables}

Similarly to the scalar case, 
for random vectors $\vec{X}$ and $\vec{Y}$, let {$\vec{X}^{[n]} = (\vec{X}^{(1)}, ..., \vec{X}^{(n)})$ 
and 
$\vec{Y}^{[m]} = (\vec{Y}^{(1)}, ..., \allowbreak \vec{Y}^{(m)})$}
be two sets of i.i.d. samples from $\vec{X}$ and $\vec{Y}$, respectively.
Then, we can define the ECDF and the complimentary ECDFs from $\vec{X}^{[n]}$ by
\begin{equation*}
\begin{split}
    F^\pi_{\vec{X}^{[n]}} (\vec{a}) = & 
    \frac{1}{n} \sum_{i = 1}^n 
    \id \big( 
    \pi(\vec{X}^{(i)})_1 \leq \pi(a)_1, \ldots, \pi(\vec{X}^{(i)})_K \leq \pi(a)_K
    \big),
\end{split}
\end{equation*}
for each  $K$-dimensional alternation $\pi \in \Pi_K$ 
(given by \cref{def:alternation}). Similarly, we can define $G^\pi_{\vec{Y}^{[m]}} (\vec{a})$ from $\vec{Y}^{[m]}$.

Following~\cite{peacock_TwodimensionalGoodnessoffitTesting_1983,fasano_MultidimensionalVersionKolmogorovSmirnov_1987}, we note that generally $\nm{F^\pi_{\vec{X}^{[n]}} - G^\pi_{\vec{Y}^{[m]}}}_{\infty}$ are not equal for all $\pi \in \Pi_K$. Thus, defining the test statistics by only using the CDFs $F^\pi_{\vec{X}^{[n]}}$ and $G^\pi_{\vec{Y}^{[m]}}$  
by $\delta_{X^{[n]}, Y^{[m]}} = \nm{F_{\vec{X}^{[n]}} - G_{\vec{Y}^{[m]}}}_{\infty}$,
as in~\eqref{eq:deltanm} is not enough. Instead, the test statistics should take all the CDFs and complimentary CDFs by
\begin{equation} \label{eq:2D delta}
    \delta_{X^{[n]}, Y^{[m]}} = \max_{\pi \in \Pi_K} \nm{F^\pi_{\vec{X}^{[n]}} - G^\pi_{\vec{Y}^{[m]}}}_{\infty}.
\end{equation}
By~\cite{peacock_TwodimensionalGoodnessoffitTesting_1983,fasano_MultidimensionalVersionKolmogorovSmirnov_1987}, the test statistics $\delta_{X^{[n]}, Y^{[m]}}$ satisfies \cref{lem:independence} and asymptotically obeys the Kolmogorov-Smirnov distribution~\eqref{eq:delta_cdf}. Therefore, the statistical test~\eqref{eq:assert} extends to the multidimensional case by using $\delta_{X^{[n]}, Y^{[m]}}$ from~\eqref{eq:2D delta}. For $K \leq 3$, the confidence level for~\eqref{eq:2D delta} can be derived directly from the results of~\cite{peacock_TwodimensionalGoodnessoffitTesting_1983,fasano_MultidimensionalVersionKolmogorovSmirnov_1987}. For $K>3$, it can be computed by extending the method of~\cite{peacock_TwodimensionalGoodnessoffitTesting_1983,fasano_MultidimensionalVersionKolmogorovSmirnov_1987}.

%% file: smc.tex
\section{Statistical Verification of Probabilistic Conformance} 
\label{sec:smc}

Based on the statistical test introduced in \cref{sec:EKS}, we now propose a statistical verification algorithm to check the probabilistic conformance of two PUSs for a monotonically parametrized STL formula (as formulated in \cref{sec:formulation}). For a lucid presentation and as with most other works (e.g.,~\cite{agha_SurveyStatisticalModel_2018,larsen_StatisticalModelChecking_2016}), we focus on bounded-time properties; handling unbounded-time properties is more involving and is an avenue for future work.

Following~\cref{def:conformance}, for a monotonically parametrized STL formula $\phi_{\underline{d}}$ with $\underline{d} \in \real^K$ and for each $K$-dimensional alternation $\pi \in \Pi_K$ (from \cref{def:alternation}), let 
\begin{align} \label{eq:FG}
    \nonumber F^\pi (\underline{d}) = \pr_{\sig_1 \sim \sys_1} (\sig_1 \models \phi_{\pi(\underline{d})} ),
    \\
     G^\pi(\underline{d}) = \pr_{\sig_2 \sim \sys_2} (\sig_2 \models \phi_{\pi(\underline{d})}).
\end{align}
By the monotonicity of $\phi_{\underline{d}}$ from \cref{def:mfs}, for each for $\pi \in \Pi_K$, the multivariate functions $F^\pi$ is the CDF or a complementary CDF of the satisfaction probability of $\phi_{\underline{d}}$ for the parameter $\underline{d}$. 
The equality in~\eqref{eq:FG} is almost everywhere in Lebesgue measure, since distribution functions $F^\pi (\underline{d})$ and $G^\pi (\underline{d})$ need to be right-continuous. The same holds for $G^\pi$.

From \cref{def:conformance}, the PUSs $\sys_1$ and $\sys_2$ conform with respect to the monotonically parametrized formula $\phi_{\underline{d}}$, if the CDFs and complementary CDFs $F^\pi (\underline{d})$ and $G^\pi (\underline{d})$ are approximately equal;~i.e., 
$$\big| \pr_{\sig_1 \sim \sys_1} (\sig_1 \models \phi) - 
\pr_{\sig_2 \sim \sys_2}(\sig_2 \models \phi) \big| < c$$
if and only if
\begin{equation}
    \gamma_{\vec{X},\vec{Y}} = \max_{\pi \in \Pi_K} \nm{F^{\pi} - G^{\pi}}_\infty < c.
\end{equation}
On the other hand, this can be solved by our statistical test introduced in \cref{sec:EKS}.

Specifically, for two sets of sample paths $\sig_1^{[n]} = (\sig_1^{(1)}, ..., \sig_1^{(n)})$ 
and
$\sig_2^{[m]} = (\sig_2^{(1)}, ..., \sig_2^{(m)})$
from the PUSs $\sys_1$ and $\sys_2$, respectively,
we define the empirical approximations of 
$F(\underline{d})$ and $G(\underline{d})$~by
\begin{align} \label{eq:Fn}
    \nonumber F^{\pi}_{n} (\underline{d}) = 
    \frac{1}{n} \sum\nolimits_{i = 1}^n 
    \id ( \sig_1^{(i)} \models \phi_{\pi(\underline{d})} ),
    \\
    G^{\pi}_{m} (\underline{d}) = 
    \frac{1}{m} \sum\nolimits_{i = 1}^m 
    \id ( \sig_2^{(i)} \models \phi_{\pi(\underline{d})} ),
\end{align}
where $\pi \in \Pi_K$ and $\id(\cdot)$ is the indicator function.%
\footnote{{In the rest of this paper, we use simplified notation with subscripts $(~)_{n}$ and $(~)_{m}$ utilized  to indicate the sets of sample paths $\sig_1^{[n]}$ and $\sig_2^{[m]}$.}}
Similarly to \eqref{eq:2D delta}, the test statistics
\begin{equation} \label{eq:mD delta}
    \delta_{n,m} = \max_{\pi \in \Pi_K} \nm{F^{\pi}_{n} - G^{\pi}_{m}}_{\infty},
\end{equation}
where $\delta_{n,m}$ is the $L_\infty$ norm, satisfies \cref{lem:independence} and obeys the KS distribution from~\eqref{eq:delta_cdf} (asymptotically for $K \geq 2$); hence, the statistical test~\eqref{eq:assert} applies.
Since $F^{\pi}_{n}$ and $G^{\pi}_{m}$ are known multidimensional step functions from the samples, $\nm{F^{\pi}_{n} - G^{\pi}_{m}}_{\infty}$ is directly computable.

\cref{alg:smc} for checking probabilistic conformance terminates with probability $1$ 
and can achieve any desired confidence level $\alpha_d < 1$. The proof follows from that of \cref{thm:1}.

% \begin{algorithm}[!t]
%     \caption{Statistical verification for conformance.\label{alg:smc}}
%     \begin{algorithmic}[1]
%     \Require Desired confidence level $\alpha_d$, threshold $c > 0$ 
    
%     \State Sample sizes $n, m \gets 0$, $\alpha \gets 0$ 
%     \While{$\alpha < \alpha_d$}
        
%     \State Draw new samples from $\sys_1, \sys_2$ and update $n, m$.
    
%     \State Update $F^{\pi}_{n}, G^{\pi}_{m}$ by~\eqref{eq:Fn} and compute $\delta_{n,m}$ by~\eqref{eq:mD delta}.  
    
%     \State Update $\alpha$ by~\cite{peacock_TwodimensionalGoodnessoffitTesting_1983,fasano_MultidimensionalVersionKolmogorovSmirnov_1987}.

%     \EndWhile
    
%     \State \Return $\assert$ by~\eqref{eq:assert}.
    
%     \end{algorithmic}
% \end{algorithm}

%% file: eval.tex
\section{Case Studies and Evaluation}
\label{sec:eval}

To demonstrate the applicability of our statistical verification algorithms, %We evaluated our statistical verification algorithms 
we evaluated them on three CPS benchmarks with complex dynamics from a wide range of application domains: (1)~Toyota Powertrain, (2)~Lane-Keeping Assistant (LKA) Controllers, and (3)~$100kW$ Grid-Connected Photo Voltaic (PV) Array (due to space constraints the results of the $3^{rd}$ case study are presented in the Appendix). %, to demonstrate the utility of our statistical verification method.
We find the case study in \cref{sub:mpc} particularly important since (LKA controllers), to the best of our knowledge, previously there are few comparative studies between NN-based and conventional techniques in cyber-physical and embedded systems.

The Toyota powertrain model is derived from \cite{jin_PowertrainControlVerification_2014}.
The LKA is implemented in MATLAB using the MPC, Deep Learning, and Reinforcement Learning Toolboxes~\cite{LKA-ref}.
The PV Array is implemented using the Simscape Power Systems toolbox~\cite{simpowerPV}.
All implementations are available at \cite{simulink}.

Evaluations are performed on a laptop with Intel Xeon E-2176M CPU @ $2.7$GHz and 16 GB RAM.
For each case study, we run~\cref{alg:smc} with different indifference parameter $c$ and desired confidence level $\alpha_d$ (i.e., the probability for~\cref{alg:smc} to return the correct assertion is at least $\alpha_d$). We report the test statistics $\delta_{n,m}$,
the number of samples, total algorithm execution time, and the assertion $\assert$ when the algorithm terminates.

\subsection{Toyota Powertrain}
\label{sub:powertrain}

\begin{figure}[!t]
\centering
\includegraphics[width = 0.44\textwidth]{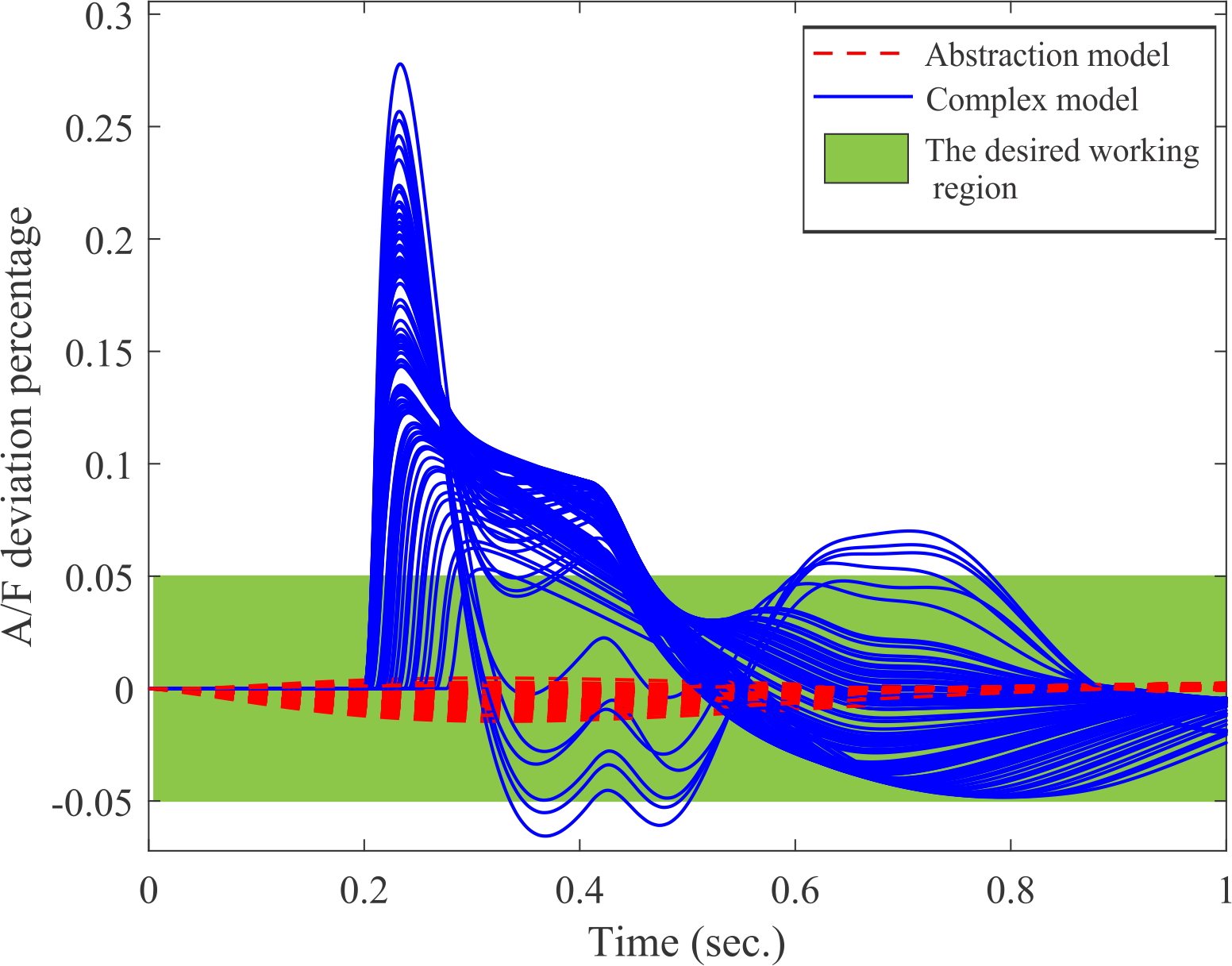}
\caption{{Sample paths from the complex (solid blue) and abstracted (dashed red) models for the A/F ratio deviation percentage. 
The paths remain inside the desired working region (in green) with a significantly higher probability for the abstracted model, illustrating that the distribution of the arrival times to the desired working region is very different for these two models. }}
\label{fig:1}
% \vspace{-10pt}
\end{figure}

We use the Simulink~models for the Toyota Powertrain with a four-mode embedded controller and $15$ state variables from~\cite{jin_PowertrainControlVerification_2014}. It is challenging to show that complex embedded/CPS with hybrid dynamics, such as the powertrain, satisfy strict performance requirements. On the one hand, the available benchmark model must capture a reasonable portion of behaviors of the real powertrain to enable  us to assess, evaluate, and verify the designs against requirements. On the other hand, the simulation time for a simpler model that sufficiently conforms with the real system is significantly~lower.

In~\cite{jin_PowertrainControlVerification_2014}, two models of the Toyota Powertrain are presented. A \emph{detailed} but complex model contains the air-to-fuel (A/F) ratio controller and an average model of the engine dynamics, such as the throttle and intake manifold air dynamics. Due to the complexity of this detailed model and limitations of existing verification tools, in~\cite{jin_PowertrainControlVerification_2014}, a simpler \emph{abstract} model as a hybrid I/O automaton is also introduced to facilitate system analysis, including formal verification.

\paragraph*{Conformance}
For the Toyota powertrain, the A/F ratio control problem is of key interest. Hence, we study the conformance for the A/F deviations $e_{A/F}$ for the \emph{detailed} and \emph{abstract} models for an RPM of $1600$ (the system input).
When the nominal input RPM is subject to Gaussian system noise $\mathbf{N} (0,18^2)$, (samples of) the change of $e_{A/F}$ over time for the two models are given in \cref{fig:1}.
The conformance requires that, under this system noise, the A/F deviations $e_{A/F}$ of the \emph{detailed} and \emph{abstract} models enter some desired working region ($|e_{A/F}| < 0.05$) in any time interval $[0.22,\tau]$ with approximately the same probability; i.e., the STL specification $\F_{[0.22,\tau]}( |e_{A/F}| < 0.05 )$ holds with approximately the same probability for any $\tau$ between the two models, as formally captured below%
\footnote{More precisely, for any $\tau \geq 0.22$ from \eqref{eq:pt}. Otherwise, the satisfaction probability is trivially $0$.}
\begin{equation}
\label{eq:pt}
\begin{split}
\forall \tau \geq 0. \
& \pr_{{\sig_a} \sim \sys_a} \big({\sig_a} \models \F_{[0.22,\tau]}( |e_{A/F}| < 0.05 )\big)
\\ & \approx_c \pr_{{\sig_f} \sim \sys_f}\big({\sig_f} \models \F_{[0.22,\tau]}( |e_{A/F}| < 0.05 )\big).
\end{split}
\end{equation}
Here, the constant $c > 0$, the approximate equality $\approx_c$ means the difference is less than $c$, the subscripts $f$ and $a$ stand for the complex and abstracted models, respectively, $e_{A/F}$ is the percentage deviation of A/F ratio, and $\tau$ is the time-bound.

\paragraph*{Result Analysis} 
We analyzed \eqref{eq:pt} using \cref{alg:smc} with the confidence level $\alpha \in \{0.95, 0.99\}$ and the conformance parameter $c \in \{0.2, \allowbreak 0.15, 0.10, 0.05\}$ (see \cref{tb:1}). 
The results are derived with relatively small numbers of samples for all confidence and indifference parameters.
The results indicate that the two employed models do not conform for the requirement~\eqref{eq:pt}, although it is claimed in~\cite{jin_PowertrainControlVerification_2014} that the abstract model is a representative of the detailed model.
Starting from the same initial RPM values, the A/F ratio in the complex  model would take more time to reach the desired working region than in most  cases in the abstracted model. 
This also agrees with \cref{fig:1}, as the A/F ratio of the abstracted model would remain inside the desired area, while in the complex model, this value exceeds the desired region in most of the cases.
Furthermore, from \cref{tb:1}, the value of the test statistics $\delta_{n,m} $ is almost $1$ in all the cases, when \cref{alg:smc} terminates.
This implies that for the detailed and abstracted models, the distribution of the startup time for their A/F ratio to reach the working region are very different --- this agrees with the algorithm assertion.

\begin{table}[!t]
\centering
\setlength\tabcolsep{2mm}
\begin{tabular}{llcllc}
	$c$ & $\alpha_d$ & $\delta_{n,m} $ & Samples &  Time~(sec.)&$\assert$ \\
	\toprule
	0.40 &0.99 &1.00 &3.9e+01 &1.8e-02 &$\False$ \\
    0.40 &0.95 &1.00 &1.9e+01 &4.4e-03 &$\False$ \\
    0.25 &0.99 &1.00 &2.5e+01 &4.6e-03 &$\False$ \\
    0.25 &0.95 &1.00 &1.3e+01 &2.2e-03 &$\False$ \\
    0.10 &0.99 &1.00 &1.8e+01 &3.6e-03 &$\False$ \\
    0.10 &0.95 &1.00 &9.0e+00 &1.6e-03 &$\False$ \\
    0.05 &0.99 &1.00 &1.6e+01 &2.8e-03 &$\False$ \\
    0.05 &0.95 &1.00 &8.0e+00 &1.3e-03 &$\False$ \\
	\bottomrule
\end{tabular}
\caption{{Statistical verification results of the conformance property~\eqref{eq:pt} 
and the test statistics $\delta_{n,m} $ upon \cref{alg:smc} termination for different values of conformance parameter $c$ and desired confidence level $\alpha_d$.} \label{tb:1}}
\vspace{-10pt}
\end{table}

\subsection{Replacing MPC with NN-based~Controllers}
\label{sub:mpc}

The controller of the LKA system is commonly based on 
model predictive control (MPC) or more recently neural networks (NN).
The conventional MPC-based controllers solve a constrained quadratic programming optimization problem from the observed state of a plant in an open-loop fashion.
This approach is usually computationally ineffective in realtime. 
Recently, NN-based controllers are employed to imitate the control rules of the MPC-based controller from samples to improve realtime computation efficiency. 
In this case study, we check the conformance of an NN-based controller and an MPC-based controller for the LKA system in MATLAB/Simulink~\cite{LKA-ref}. 

\begin{figure*}[t]
\centering
\includegraphics[width = 1.5\columnwidth]{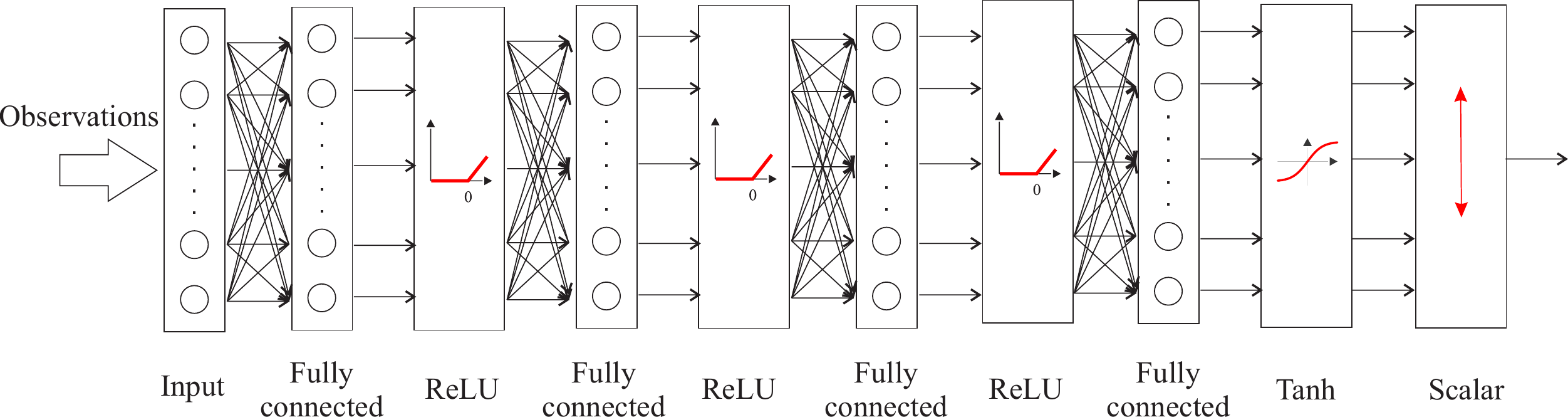}
\caption{The employed structure for the NN controllers.}
\label{fig:2}
\end{figure*}

In the LKA system, the sensors measure the lateral deviation, relative yaw angle between the center-line of a lane and the vehicle, current lane curvature, and its derivative. 
The objective of the controller is to keep the lateral error and relative yaw angle close to zero.
To dynamics of the vehicle is given by the three Degrees-of-Freedom (DoF) bicycle model~\cite{gillespie1992fundamentals}~as
\begin{align}
\nonumber\begin{bmatrix}\dot{V}_y \\\ddot{\psi}\end{bmatrix} =& \begin{bmatrix}
-\frac{2C_f+2C_r}{m V_x} & -V_x-\frac{2C_fl_f-2C_r l_r}{m V_x} \\
     -\frac{2C_f l_f-2C_r l_r}{I_z V_x} & -\frac{2 C_f l_f^2+2C_r l_r^2}{I_z V_x}
\end{bmatrix} \begin{bmatrix}V_y\\\dot{\psi}\end{bmatrix} + 2\begin{bmatrix}  \frac{C_f}{m}\\  \frac{C_f l_f}{I_z}\end{bmatrix} u \\
\nonumber y = & \begin{bmatrix}V_y & \dot{\psi}\end{bmatrix}^\mathrm{T}.
\end{align}
Here, $V_x$ is the longitudinal velocity, $m$ is the total vehicle mass, $I_z$ is the yaw moment of inertia of the vehicle, $l_f$ and $l_r$ are the longitudinal distance from the center of gravity to the front and real tires, and $C_f$ and $C_r$ are the cornering stiffness of the front and rear tires, respectively. The system state consist of the lateral 
velocity $V_y$ and yaw angle rate $\dot{\psi}$, and the front steering angle $u(t)$ is the system input.

\paragraph*{MPC} 
The MPC-based controller is derived from the MPC 
toolbox in MATLAB.
The values of the variables are set as follows: 
$V_x = 15 \ m/s$, 
$m = 1575 \ kg$, 
$I_z = 2875 \ m \cdot N \cdot s^2$, 
$l_f = 1.2 \ m$, 
$l_r = 1.6 \ m$, 
$C_f = 19000 \ N/\mathrm{rad}$, 
and $C_r=33000 \ N/\mathrm{rad}$. 
The controller output is confined within the interval $[-\pi/3,\pi/3] \ \mathrm{rad}$. 
The predictive time horizon and control time horizon are set to $h_p = 20$ and $h_c = 20$.

\paragraph*{DNN Replacement} 
We train a NN controller to replace the 
MPC controller,
by sampling from the MPC based controller 
for randomly generated states, 
last control action, and measured disturbances. 
The samples are divided into the train and validation testing data,
and are used to train several NNs with similar structure, 
but different numbers of neurons per 
layer (30, 45, 60, and 300 neurons per layer). 
All middle layers are fully 
connected with ReLU activation functions and the output layer is a fully-connected layer with $\tanh$ activation function and a scalar layer. 
The maximal number of epoch to stop the training is set to $30$. 
The structure of the NNs is shown in \cref{fig:2}.

\paragraph*{Conformance}

For the input of the same reference path of the vehicle (given by the Matlab Toolbox), we expect that using the NN controller the lateral deviation of the vehicle under random values of the initial states should be similar to the output of the MPC-based closed-loop system.
Thus, we assign an upper bound to the error of the lateral deviation and check when the designed controller reaches this boundary. 
With fixed values of initial states, we run the closed-loop system with two NNs and the reference MPC. 
Then, we compare the time that the absolute value of the lateral deviation falls below the desired value for the NN controller and the MPC controller;
this is formally captured by the STL formula
$\F_{[0,\tau]}(|e_y| < \gamma )$ 
monotonically parametrized by $\tau$.
Accordingly, the conformance between the MPC-controlled and NN-controlled
LKA systems for this parametrized specification is
\begin{equation}\label{eq:phi2}
\begin{split}
\forall \tau \geq 0. \ 
& \pr_{\sig_1 \sim \sys_{\mathsf{NN}}} (\sig_1 \models 
\F_{[0,\tau]}(|e_y^{\mathsf{NN}}| < \gamma)) 
\\ & 
\approx_c
\pr_{\sig_2 \sim \sys_{\mathsf{MPC}}}({\sig_2} 
\models \F_{[0,\tau]}(|e_y^{\mathsf{MPC}}| < \gamma)),
\end{split}
\end{equation}
where the constants $c, \gamma > 0$,
the approximate equality $\approx_c$ means the difference is less than $c$,
and $e_y$ is the lateral deviation of the intended controller. 
The random signals $\sig_1$ and $\sig_2$ are derived as follows. The initial conditions of the system such as the lateral velocity $V_y$, yaw angle rate $\dot{\psi}$, lateral deviation $e_1$, relative yaw angle $e_2$, last 
steering angle $u$, and the measured road yaw rate $V_x \rho$ are drawn randomly using the uniform distribution from intervals 
$[-2, 2] \ m/s$, 
$[-\pi/3, \pi/3] \ {\mathrm{rad}}/s$, 
$[-1,1]m$, 
$[-\pi/4, \pi/4] \ {\mathrm{rad}}$, 
$[-\pi/3, \pi/3] \ {\mathrm{rad}}$, 
and $[-0.01, 0.01]$, 
respectively. 
The minimum road reduce is $100 \ m$.

\paragraph*{Result Analysis} 
The results for applying~\cref{alg:smc} with parameters $\alpha\in\{0.95, 0.99\}$, and $c \in \{0.40,  0.25, \allowbreak 0.10, 0.05\}$ are shown in \cref{tb:2} for NN controllers with $30$ and $45$ neurons per layer. As can be seen, the NN controllers with $45$ neurons per layer conforms much better with the MPC controller than the NN controllers with $30$ neurons per layer for the requirement~\eqref{eq:phi2}.
The results for $60$ and $300$ neurons per layer are similar to $45$ neurons per layer (as confirmed by \cref{fig:nn}), 
so they are omitted due to the space limit. 
All these results are achieved with relatively few samples (at most a few thousand samples for each~setup).
\begin{table*}[!t]
	\centering
	\setlength\tabcolsep{2mm}
	\begin{tabular}{crcrcc||crcc}
	   \multicolumn{6}{c}{NN (30 Neurons per Layer)} & \multicolumn{4}{c}{NN (45 Neurons per Layer)}\\
	   \toprule
		$c$ & $\alpha_d$ &$\delta_{n,m} $ & Samples & $T(s)$ &
$\assert$ & $\delta_{n,m} $ &Samples & $T(s)$ &$\assert$\\
		\toprule
		0.40 &0.99 &0.98 &4.3e+01 &7.4e-03 &$\False$ &0.36 &1.0e+04 &9.6e+00 &$\True$ \\
        0.40 &0.95 &1.00 &1.9e+01 &3.1e-03 &$\False$ &0.36 &3.6e+03 &2.0e+00 &$\True$ \\
        0.25 &0.99 &1.00 &2.5e+01 &4.1e-03 &$\False$ &0.37 &9.5e+02 &3.2e-01 &$\False$ \\
        0.25 &0.95 &1.00 &1.3e+01 &2.1e-03 &$\False$ &0.42 &2.5e+02 &5.9e-02 &$\False$ \\
        0.10 &0.99 &1.00 &1.8e+01 &3.0e-03 &$\False$ &0.36 &2.1e+02 &4.2e-02 &$\False$ \\
        0.10 &0.95 &1.00 &9.0e+00 &1.6e-03 &$\False$ &0.35 &1.2e+02 &2.2e-02 &$\False$ \\
        0.05 &0.99 &1.00 &1.6e+01 &2.7e-03 &$\False$ &0.38 &1.3e+02 &2.5e-02 &$\False$ \\
        0.05 &0.95 &1.00 &8.0e+00 &1.2e-03 &$\False$ &0.36 &7.3e+01 &1.4e-02 &$\False$ \\
		\bottomrule
	\end{tabular}
\caption{Statistical verification results 
for the conformance property \eqref{eq:phi2} 
and the test statistics $\delta_{n,m} $ upon \cref{alg:smc} termination for different values of the conformance parameter $c$ and desired confidence level $\alpha_d$.}
\label{tb:2}
\vspace{-16pt}
\end{table*}

The results of \cref{tb:2} imply that increasing the size of the NN-based controllers improves the conformance with the MPC controller. To check this observation and confirm the results of \cref{tb:2}, we show in \cref{fig:nn} the ECDFs of the settling time for the MPC controller and the NN controllers with $30$, $45$, $60$, and $300$ neurons per layer; the conformance for the requirement \eqref{eq:phi2} is visually demonstrated by the closeness of the ECDFs.
To derive the same conclusion, each ECDF uses $200$ samples, which is significantly more 
than the samples required by \cref{alg:smc}, as shown in~\cref{tb:2}.
As shown in~\cref{fig:nn}, increasing the number of neurons beyond 45 does not lead to considerable change in the CDF of the settling times for NN based controllers. 
Comparing to NN$_{300}$, the NN$_{60}$ controller has better conformance with the MPC. The latter implies that NN$_{300}$ controller has the over-fitting problem. 
For the NN-based controllers of different sizes, 
the test statistics upon algorithm termination is 
$\delta_{n,m} ^{\mathsf{NN}_{30}} = 0.98$, 
$\delta_{n,m} ^{\mathsf{NN}_{45}} = 0.31$, 
$\delta_{n,m} ^{\mathsf{NN}_{60}} = 0.31$,
and $\delta_{n,m} ^{\mathsf{NN}_{300}} = 0.35$.

\begin{figure}[!t]
\centering
\includegraphics[width = 0.9\columnwidth]{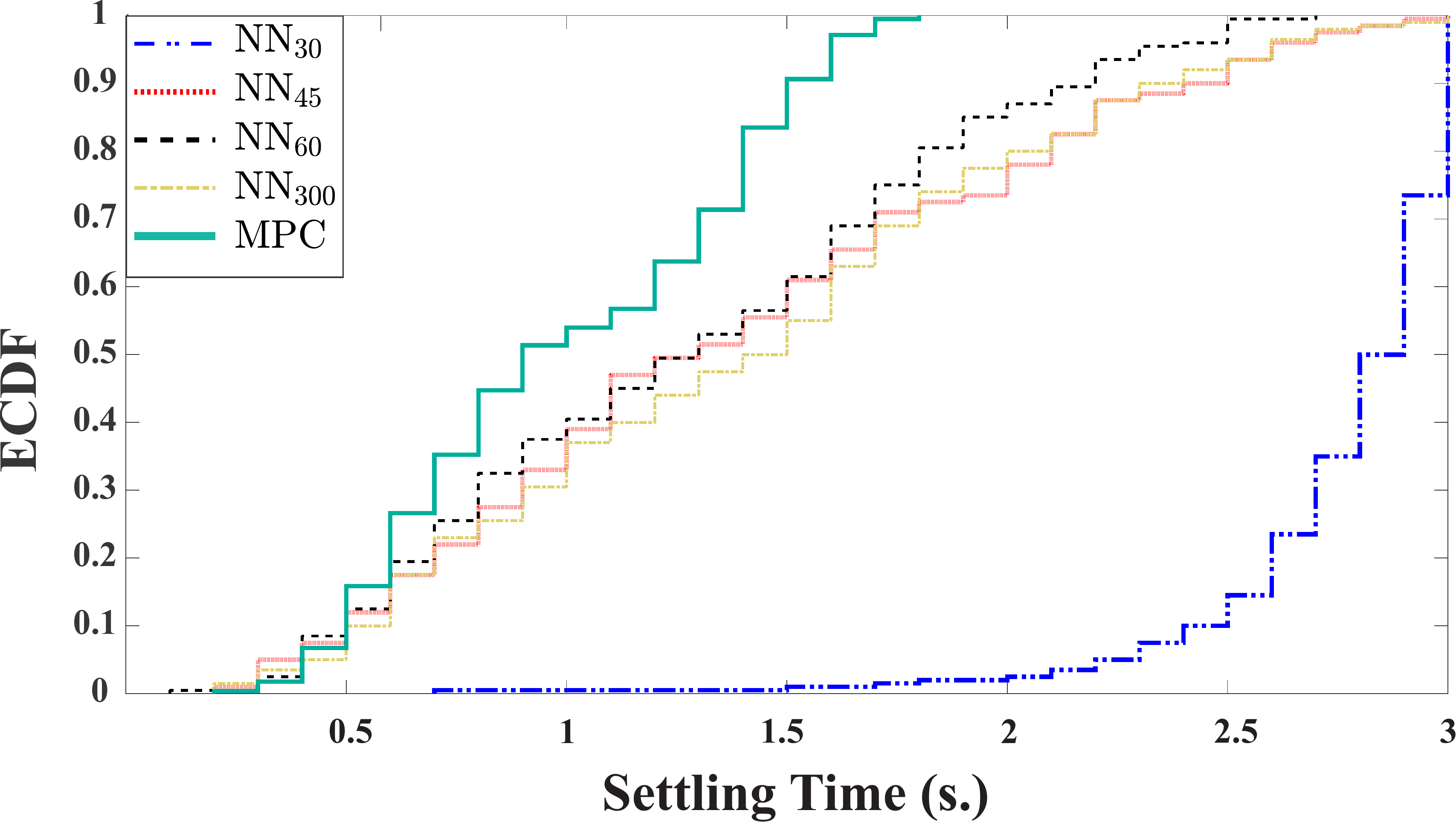}
\caption{{The ECDFs of the settling time for the MPC and NN-based controllers with with $30$, $45$, $60$, and $3000$ neurons per layer (NN$_{30}$, NN$_{45}$, NN$_{60}$, and NN$_{300}$) from $200$ samples. The conformance is visually demonstrated by the closeness of the ECDFs.}}
\label{fig:nn}
% \vspace{-17pt}
\end{figure}

%% file: related.tex
\section{Related Work} 
\label{sec:related}

Conformance of CPS for different types of specifications 
of interest is studied in~\cite{ryabtsev2009translation,
majumdar2013compositional,
khakpour_NotionsConformanceTesting_2015,
roehm_ReachsetConformanceTesting_2016,
liu_ReachsetConformanceForward_2018,
graf2019component}.
As in~\cite{abbas_FormalPropertyVerification_2014,
deshmukh_QuantifyingConformanceUsing_2017},
in this work, we focus on a class of conformance properties for CPS that are specified by temporal logic formulas.
Our notion of conformance can be viewed as the 
probabilistic extension of 
\cite{abbas_FormalPropertyVerification_2014,
deshmukh_QuantifyingConformanceUsing_2017},
that is needed to allow for capturing the conformance between a wide class of 
probabilistic CPS (which we model as PUSs).
Since reachability properties can  be in general captured by
temporal logic formulas, our notion of conformance 
is more general than 
the conformance for reachability from~\cite{roehm_ReachsetConformanceTesting_2016,
liu_ReachsetConformanceForward_2018}. %\todo{check}

Existing works on conformance 
for temporal logic specifications
mainly focus on non-probabilistic models~\cite{abbas_FormalPropertyVerification_2014,khakpour_NotionsConformanceTesting_2015,graf2019component,deshmukh_QuantifyingConformanceUsing_2017,liu_ReachsetConformanceForward_2018}. 
On the other hand, in this work, we focus on a probabilistic notion of 
conformance -- the satisfaction probability of the specifications of interest
should be approximately equal.
In~\cite{abbas_FormalPropertyVerification_2014,deshmukh_QuantifyingConformanceUsing_2017}, 
conformance builds a relation between two models 
such that if \emph{any} STL formula holds on one model, then the corresponding formula should automatically hold on the other model. 
Conceptually, our notion of conformance is less stringent, as it only involves a given set of STL formulas of interest. 
Furthermore, our notion of conformance is conceptually more general than~\cite{roehm_ReachsetConformanceTesting_2016,liu_ReachsetConformanceForward_2018},
where the conformance is only for reachability.
Our notion of conformance can specify the conformance of probabilistic 
reachability for two models.

Conformance is different from 
the simulation/bisimulation~\cite{damsAbstractionAbstractionRefinement2018} in two aspects.
Conceptually, conformance focuses on the level of functionality, and only captures the similarity between two models for a set of specifications of interest.
That is, the behavior of the two models may be very different 
for other specifications (not of interest).
On the other hand, the simulation/bisimulation focuses 
on the level of executions, and requires an execution-wise correspondence 
between the two models.
Also, the two concepts have slightly different 
domains of applications~\cite{khakpour_NotionsConformanceTesting_2015,deshmukh_QuantifyingConformanceUsing_2017,
abbas_FormalPropertyVerification_2014}.
Conformance is commonly only used for cyber-physical and embedded control 
systems, 
while simulation/bisimulation may be used for both 
discrete models~\cite{damsAbstractionAbstractionRefinement2018} as well as 
cyber-physical and embedded 
control~systems~\cite{julius_ApproximationsStochasticHybrid_2009,
wang_StatisticalVerificationDynamical_2015}.
% \todo{what do you mean by the last sentence?}

To the best of our knowledge, this is the first work on 
statistically verifying the probabilistic conformance of CPS with complex dynamics (formally captured as probabilistic 
uncertain systems from \cref{def:PUS}), while providing provable confidence levels (i.e., false positive/negative ratios).
Existing model-based methods for conformance,
such as 
\cite{abbas_FormalPropertyVerification_2014,
khakpour_NotionsConformanceTesting_2015,
deshmukh_QuantifyingConformanceUsing_2017,
liu_ReachsetConformanceForward_2018}
cannot directly handle such systems
with complex or even unknown dynamics in practice. 
On the other hand, existing conformance testing methods
for temporal logic specifications
\cite{roehm_ReachsetConformanceTesting_2016,
graf2019component}
or other specifications~\cite{ryabtsev2009translation,majumdar2013compositional}
cannot provide probabilistic guarantees like the presented method.
Therefore, those methods are not directly comparable
with ours for the case studies presented in \cref{sec:eval}.

%% file: appendix.tex
\begin{appendix}

\section{Appendix}

\subsection{Power Plant Case-Study}
\label{sub:power}

In the final case-study, we compare the \emph{detailed} 
and \emph{average} models of a $100kW$ array connected to a $25kV$ grid via a DC-DC boost converter and a 
three-phase three-level Voltage Source Converter~(VSC), from the MATLAB Simscape Electrical Toolbox~\cite{PowerSim}. 
{Both models include a Photo Voltaic (PV) array that delivers the maximum power of $100 \ kW$ at $1000 \ W/m^2$ sun irradiance,
a DC-DC boost converter, 
3-level 3-phase VSC, 
capacitor bank, 
three-phase coupling transformer, 
and a given utility grid.}
{The models use the Simulink model of a boost converter to implement the Maximum Power Point Tracking (MPPT). The MPPT optimizes the match between the solar array (PV panels) and the utility grid. The models have differences such as employed technique to implement MPPT, DC-DC, and VSC converters' structure %. For more details on these models, the readers are encouraged to consult with~
\cite{de2011comparative}.}

The VSC converts the $500V$ DC link voltage to $260V$ AC and keeps unity 
power factor. 
To this end, two control loops are employed: one control loop 
regulates DC link voltage to $\pm250V$ (external controller) and the other 
control loop regulates active and reactive grid currents (internal controller). 
The active current reference is the output of the DC voltage external 
controller. The latter controller is a PI (proportional–integral) controller whose input is the error of  the DC voltage. 

% Thus, we focus on the deviation of the DC voltage for both models. 

\paragraph*{Conformance}

We consider the deviation of the DC voltage $e_{vdc}$, 
when the sun irradiance and environment temperature are subject to changes. 
For an arbitrary threshold~$\gamma$, we use the STL specification 
$\G_{[0.5,2]}(|e_{vdc} | < \gamma)$,
which is monotonically parametrized by $\gamma$,
to capture that $e_{vdc}$ is always below $\gamma$
within the time interval $[0.5,2]$ of interest.
Accordingly, the conformance between the detailed and average models
for this parametrized specification is captured by
\begin{equation}\label{eq:phi3}
\begin{split}
\forall \gamma \geq 0. \
& \pr_{{\sig_d} \sim \sys_d} ({\sig_d} \models \G_{[0.5,2]}( |e_{vdc_{d}} 
| < \gamma) ) 
\\ & 
\approx_c
\pr_{{\sig_a} \sim \sys_a}({\sig_a} \models 
\G_{[0.5,2]}(|e_{vdc_{a}} | < \gamma) ),
\end{split}
\end{equation}
where the constant $c > 0$,
the approximate equality $\approx_c$ means the difference is less than $c$,
and the detailed and average models are denoted by $d$ and $a$, respectively.

\begin{figure}[t]
    \centering
    \includegraphics[width = 0.9\columnwidth]{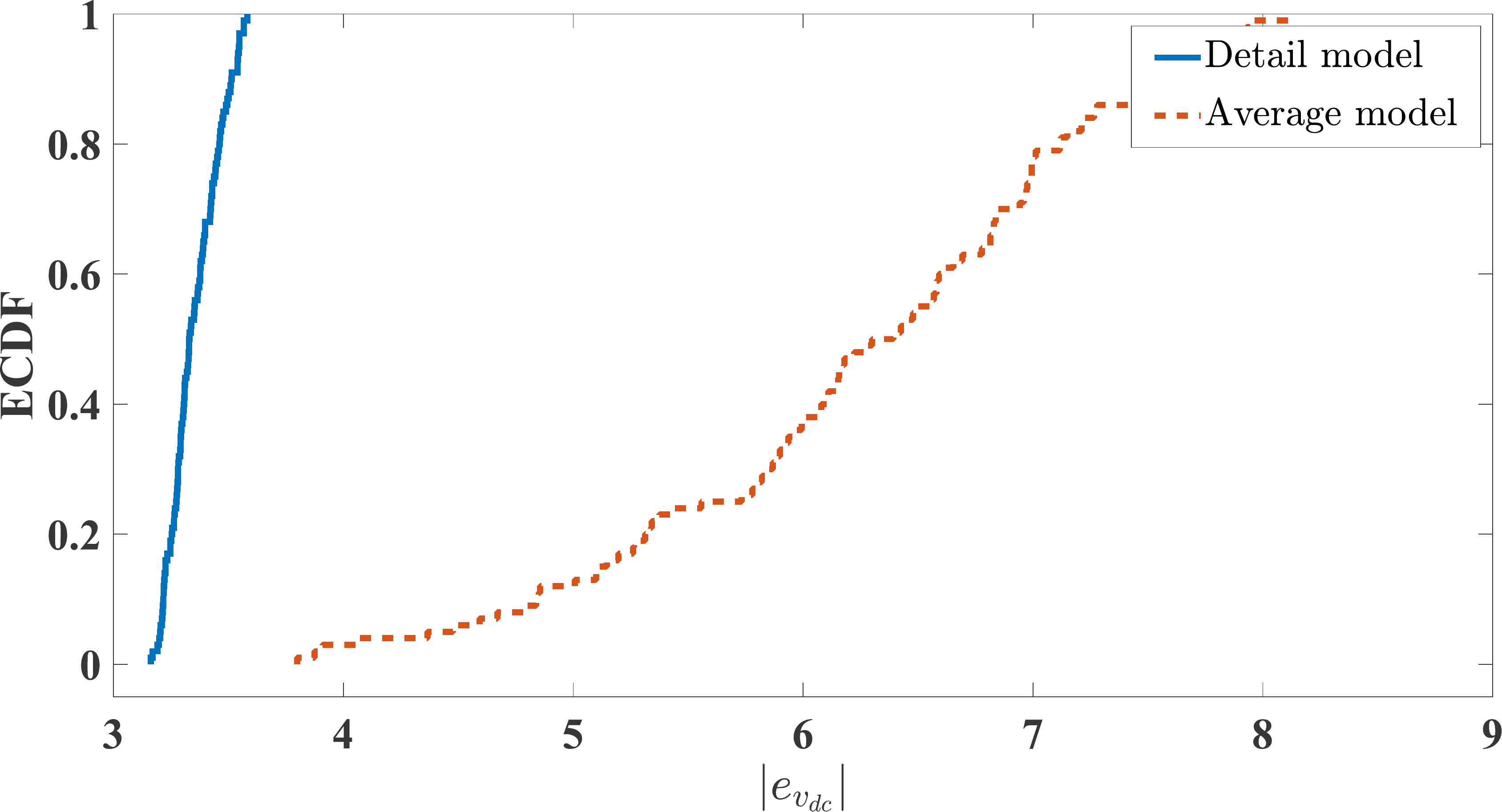}
    \caption{{The ECDFs for the maximum deviation of $V_{dc}$ in the detailed and average models for $100$ samples. The two distributions of the maximum errors for the models are noticeably~different.}}
    \label{fig:pp}
\end{figure}

We applied~\cref{alg:smc} with parameters
$\alpha\in\{0.95, 0.99\}$ and 
$c \in \{0.001, 0.005, 0.01, 0.05\}$. 
For both the models, we consider the standard test conditions 
(initial temperature and irradiance are $25^\circ$ and $1000~W/m^2$, respectively) 
with the following scenario (i.e., the input to the models):
\begin{enumerate}
    \item At $t=0.3s$ MPPT starts to regulate PV voltage.
    
    \item In time interval $[0.6, 1.1]s$, the sun irradiance linearly is ramped 
     to a minimum value. Also, the environment temperature start increasing to a 
     maximum value, simultaneously.
     
    \item In time interval $[1.1, 1.2]s$, the sun irradiance and 
    environment temperature stay constant. The minimum value of the irradiance is drawn randomly from a distribution $\mathbf{N}_{ir} (650,10^2)$
    and the maximum temperature is $20 - 0.02\times \mathbf{N}_{ir} (650,10^2)$.  
    
    \item In time interval $[1.2, 1.7]s$, the sun irradiance and 
    temperature are linearly restored back to $1000W/m^2$ and $25^\circ$, respectively; from 
    then onward, remain constant. 
\end{enumerate}

\begin{table}[!t]
\centering
\setlength\tabcolsep{2mm}
\begin{tabular}{llcllc}
$c$ & $\alpha_d$& $\delta_{n,m}$& Samples &  Time~(sec.)&$\assert$ \\
\toprule
0.40 &0.99 &1.00 &3.9e+01 &1.0e-02 &$\False$ \\
0.40 &0.95 &1.00 &1.9e+01 &6.9e-03 &$\False$ \\
0.25 &0.99 &1.00 &2.5e+01 &5.3e-03 &$\False$ \\
0.25 &0.95 &1.00 &1.3e+01 &3.3e-03 &$\False$ \\
0.10 &0.99 &1.00 &1.8e+01 &3.8e-03 &$\False$ \\
0.10 &0.95 &1.00 &9.0e+00 &1.8e-03 &$\False$ \\
0.05 &0.99 &0.94 &1.8e+01 &3.2e-03 &$\False$ \\
0.05 &0.95 &1.00 &8.0e+00 &1.3e-03 &$\False$ \\
\bottomrule
\end{tabular}
\caption{{Statistical verification results 
of the conformance property \eqref{eq:phi3} and the test statistics $\delta_{n,m}$ upon \cref{alg:smc} termination,
for different conformance parameter $c$ and desired
confidence level $\alpha_d$.}}
\label{tb:3}
\end{table}

\paragraph*{Result Analysis}  
\cref{tb:3} contains the results that demonstrate the nonconformance of the \emph{detailed} and \emph{average} models for the requirement \eqref{eq:phi3}, 
although it is commonly believed that 
the average model is generally a good approximation of the detailed model
\cite{PowerSim}.
This result is achieved with a relatively small number of samples (at most a few dozen samples for each setup).
The results for the considered specification reveals that two models do not have conformance for any values of $c$. 
To confirm the results of \cref{tb:3},
we plot in \cref{fig:pp} the ECDFs of 
the maximum deviation $|e_{V_{dc}}|$ of 
the detailed and average models;
the discrepancy of the two ECDFs demonstrates 
the nonconformance of two models for the requirement \eqref{eq:phi3}.
Each ECDF uses $100$ samples, which is significantly more 
than the samples required by \cref{alg:smc}, as shown by \cref{tb:3}.

\end{appendix}

%% file: main.bbl

\begin{thebibliography}{42}

%%% the Web address.

\ifx \showCODEN    \undefined \def \showCODEN     #1{\unskip}     \fi
\ifx \showDOI      \undefined \def \showDOI       #1{#1}\fi
\ifx \showISBNx    \undefined \def \showISBNx     #1{\unskip}     \fi
\ifx \showISBNxiii \undefined \def \showISBNxiii  #1{\unskip}     \fi
\ifx \showISSN     \undefined \def \showISSN      #1{\unskip}     \fi
\ifx \showLCCN     \undefined \def \showLCCN      #1{\unskip}     \fi
\ifx \shownote     \undefined \def \shownote      #1{#1}          \fi
\ifx \showarticletitle \undefined \def \showarticletitle #1{#1}   \fi
\ifx \showURL      \undefined \def \showURL       {\relax}        \fi
% invisible to TeX
\providecommand\bibfield[2]{#2}
\providecommand\bibinfo[2]{#2}
\providecommand\natexlab[1]{#1}
\providecommand\showeprint[2][]{arXiv:#2}

\bibitem[\protect\citeauthoryear{Abbas, Mittelmann, and Fainekos}{Abbas
  et~al\mbox{.}}{2014}]%
        {abbas_FormalPropertyVerification_2014}
\bibfield{author}{\bibinfo{person}{Houssam Abbas}, \bibinfo{person}{Hans
  Mittelmann}, {and} \bibinfo{person}{Georgios Fainekos}.}
  \bibinfo{year}{2014}\natexlab{}.
\newblock \showarticletitle{Formal Property Verification in a Conformance
  Testing Framework}. In \bibinfo{booktitle}{\emph{{{12th ACM}}/{{IEEE Conf.}}
  on {{Formal Methods}} and {{Models}} for {{Codesign}} ({{MEMOCODE}})}}.
  \bibinfo{pages}{155--164}.
\newblock


\bibitem[\protect\citeauthoryear{Agha and Palmskog}{Agha and Palmskog}{2018}]%
        {agha_SurveyStatisticalModel_2018}
\bibfield{author}{\bibinfo{person}{Gul Agha} {and} \bibinfo{person}{Karl
  Palmskog}.} \bibinfo{year}{2018}\natexlab{}.
\newblock \showarticletitle{A {{Survey}} of {{Statistical Model Checking}}}.
\newblock \bibinfo{journal}{\emph{ACM Trans. Model. Comput. Simul.}}
  \bibinfo{volume}{28}, \bibinfo{number}{1} (\bibinfo{year}{2018}),
  \bibinfo{pages}{6:1--6:39}.
\newblock


\bibitem[\protect\citeauthoryear{Asarin, Donz{\'e}, Maler, and Nickovic}{Asarin
  et~al\mbox{.}}{2011}]%
        {asarin2011parametric}
\bibfield{author}{\bibinfo{person}{Eugene Asarin}, \bibinfo{person}{Alexandre
  Donz{\'e}}, \bibinfo{person}{Oded Maler}, {and} \bibinfo{person}{Dejan
  Nickovic}.} \bibinfo{year}{2011}\natexlab{}.
\newblock \showarticletitle{Parametric identification of temporal properties}.
  In \bibinfo{booktitle}{\emph{International Conference on Runtime
  Verification}}. \bibinfo{pages}{147--160}.
\newblock


\bibitem[\protect\citeauthoryear{Barrett, Speth, Eastham, Dedoussi, Ashok,
  Malina, and Keith}{Barrett et~al\mbox{.}}{2015}]%
        {barrett2015impact}
\bibfield{author}{\bibinfo{person}{Steven~RH Barrett},
  \bibinfo{person}{Raymond~L Speth}, \bibinfo{person}{Sebastian~D Eastham},
  \bibinfo{person}{Irene~C Dedoussi}, \bibinfo{person}{Akshay Ashok},
  \bibinfo{person}{Robert Malina}, {and} \bibinfo{person}{David~W Keith}.}
  \bibinfo{year}{2015}\natexlab{}.
\newblock \showarticletitle{Impact of the Volkswagen emissions control defeat
  device on US public health}.
\newblock \bibinfo{journal}{\emph{Environmental Research Letters}}
  \bibinfo{volume}{10}, \bibinfo{number}{11} (\bibinfo{year}{2015}),
  \bibinfo{pages}{114005}.
\newblock


\bibitem[\protect\citeauthoryear{Barthe, D’Argenio, Finkbeiner, and
  Hermanns}{Barthe et~al\mbox{.}}{2016}]%
        {barthe_FacetsSoftwareDoping_2018}
\bibfield{author}{\bibinfo{person}{Gilles Barthe}, \bibinfo{person}{Pedro~R
  D’Argenio}, \bibinfo{person}{Bernd Finkbeiner}, {and}
  \bibinfo{person}{Holger Hermanns}.} \bibinfo{year}{2016}\natexlab{}.
\newblock \showarticletitle{Facets of software doping}. In
  \bibinfo{booktitle}{\emph{International Symposium on Leveraging Applications
  of Formal Methods}}. Springer, \bibinfo{pages}{601--608}.
\newblock


\bibitem[\protect\citeauthoryear{{CPSL@Duke}}{{CPSL@Duke}}{2020}]%
        {simulink}
\bibfield{author}{\bibinfo{person}{{CPSL@Duke}}.}
  \bibinfo{year}{2020}\natexlab{}.
\newblock \bibinfo{title}{Probabilistic Conformance for CPS: Case-Studies}.
\newblock
  \bibinfo{howpublished}{\url{https://gitlab.oit.duke.edu/cpsl/conformance}}.
\newblock


\bibitem[\protect\citeauthoryear{Dams and Grumberg}{Dams and Grumberg}{2018}]%
        {damsAbstractionAbstractionRefinement2018}
\bibfield{author}{\bibinfo{person}{Dennis Dams} {and} \bibinfo{person}{Orna
  Grumberg}.} \bibinfo{year}{2018}\natexlab{}.
\newblock \showarticletitle{Abstraction and Abstraction Refinement}.
\newblock In \bibinfo{booktitle}{\emph{Handbook of Model Checking}}.
  \bibinfo{publisher}{{Springer International}}, \bibinfo{pages}{385--419}.
\newblock


\bibitem[\protect\citeauthoryear{De~Brito, Sampaio, Luigi, e~Melo, and
  Canesin}{De~Brito et~al\mbox{.}}{2011}]%
        {de2011comparative}
\bibfield{author}{\bibinfo{person}{Moacyr A.~G. De~Brito},
  \bibinfo{person}{Leonardo~P. Sampaio}, \bibinfo{person}{G. Luigi},
  \bibinfo{person}{Guilherme~A. e Melo}, {and} \bibinfo{person}{Carlos~A.
  Canesin}.} \bibinfo{year}{2011}\natexlab{}.
\newblock \showarticletitle{Comparative analysis of {MPPT} techniques for {PV}
  applications}. In \bibinfo{booktitle}{\emph{2011 International Conference on
  Clean Electrical Power}}. \bibinfo{pages}{99--104}.
\newblock


\bibitem[\protect\citeauthoryear{Deshmukh, Majumdar, and Prabhu}{Deshmukh
  et~al\mbox{.}}{2017}]%
        {deshmukh_QuantifyingConformanceUsing_2017}
\bibfield{author}{\bibinfo{person}{Jyotirmoy~V. Deshmukh},
  \bibinfo{person}{Rupak Majumdar}, {and} \bibinfo{person}{Vinayak~S. Prabhu}.}
  \bibinfo{year}{2017}\natexlab{}.
\newblock \showarticletitle{Quantifying Conformance Using the {{Skorokhod}}
  Metric}.
\newblock \bibinfo{journal}{\emph{Formal Methods in System Design}}
  \bibinfo{volume}{50}, \bibinfo{number}{2} (\bibinfo{year}{2017}),
  \bibinfo{pages}{168--206}.
\newblock


\bibitem[\protect\citeauthoryear{Deshpande, {Naik-Nimbalkar}, and
  Dewan}{Deshpande et~al\mbox{.}}{2018}]%
        {deshpande_NonparametricStatisticsTheory_2018}
\bibfield{author}{\bibinfo{person}{Jyotirmoy~V. Deshpande},
  \bibinfo{person}{Uttara {Naik-Nimbalkar}}, {and} \bibinfo{person}{Isha
  Dewan}.} \bibinfo{year}{2018}\natexlab{}.
\newblock \bibinfo{booktitle}{\emph{Nonparametric Statistics: Theory and
  Methods}}.
\newblock


\bibitem[\protect\citeauthoryear{Fasano and Franceschini}{Fasano and
  Franceschini}{1987}]%
        {fasano_MultidimensionalVersionKolmogorovSmirnov_1987}
\bibfield{author}{\bibinfo{person}{Giovanni Fasano} {and}
  \bibinfo{person}{Alberto Franceschini}.} \bibinfo{year}{1987}\natexlab{}.
\newblock \showarticletitle{A multidimensional version of the
  Kolmogorov--Smirnov test}.
\newblock \bibinfo{journal}{\emph{Monthly Notices of the Royal Astronomical
  Society}} \bibinfo{volume}{225}, \bibinfo{number}{1} (\bibinfo{year}{1987}),
  \bibinfo{pages}{155--170}.
\newblock


\bibitem[\protect\citeauthoryear{Gillespie}{Gillespie}{1992}]%
        {gillespie1992fundamentals}
\bibfield{author}{\bibinfo{person}{Thomas~D Gillespie}.}
  \bibinfo{year}{1992}\natexlab{}.
\newblock \bibinfo{booktitle}{\emph{Fundamentals of vehicle dynamics}}.
  Vol.~\bibinfo{volume}{400}.
\newblock \bibinfo{publisher}{Society of automotive engineers Warrendale, PA}.
\newblock


\bibitem[\protect\citeauthoryear{Graf-Brill and Hermanns}{Graf-Brill and
  Hermanns}{2019}]%
        {graf2019component}
\bibfield{author}{\bibinfo{person}{Alexander Graf-Brill} {and}
  \bibinfo{person}{Holger Hermanns}.} \bibinfo{year}{2019}\natexlab{}.
\newblock \showarticletitle{Component-aware Input-Output Conformance}. In
  \bibinfo{booktitle}{\emph{International Conference on Formal Techniques for
  Distributed Objects, Components, and Systems}}. \bibinfo{pages}{111--128}.
\newblock


\bibitem[\protect\citeauthoryear{Heerink and Tretmans}{Heerink and
  Tretmans}{1996}]%
        {heerink_FormalMethodsConformance_1996}
\bibfield{author}{\bibinfo{person}{Lex Heerink} {and} \bibinfo{person}{Jan
  Tretmans}.} \bibinfo{year}{1996}\natexlab{}.
\newblock \showarticletitle{Formal Methods in Conformance Testing: A
  Probabilistic Refinement}.
\newblock In \bibinfo{booktitle}{\emph{{{Int. Work.}} on {{Testing}} of
  {{Communicating Sys.}}}} \bibinfo{pages}{261--276}.
\newblock


\bibitem[\protect\citeauthoryear{Henzinger}{Henzinger}{2000}]%
        {henzinger2000theory}
\bibfield{author}{\bibinfo{person}{Thomas~A Henzinger}.}
  \bibinfo{year}{2000}\natexlab{}.
\newblock \showarticletitle{The theory of hybrid automata}.
\newblock In \bibinfo{booktitle}{\emph{Verification of Digital and Hybrid
  Systems}}. \bibinfo{publisher}{Springer}, \bibinfo{pages}{265--292}.
\newblock


\bibitem[\protect\citeauthoryear{Jin, Deshmukh, Kapinski, Ueda, and Butts}{Jin
  et~al\mbox{.}}{2014}]%
        {jin_PowertrainControlVerification_2014}
\bibfield{author}{\bibinfo{person}{Xiaoqing Jin}, \bibinfo{person}{Jyotirmoy~V.
  Deshmukh}, \bibinfo{person}{James Kapinski}, \bibinfo{person}{Koichi Ueda},
  {and} \bibinfo{person}{Ken Butts}.} \bibinfo{year}{2014}\natexlab{}.
\newblock \showarticletitle{Powertrain Control Verification Benchmark}. In
  \bibinfo{booktitle}{\emph{The 17th International Conference on {{Hybrid}}
  Systems: Computation and Control}}. \bibinfo{pages}{253--262}.
\newblock


\bibitem[\protect\citeauthoryear{Julius and Pappas}{Julius and Pappas}{2009}]%
        {julius_ApproximationsStochasticHybrid_2009}
\bibfield{author}{\bibinfo{person}{Augung~A. Julius} {and}
  \bibinfo{person}{George~J. Pappas}.} \bibinfo{year}{2009}\natexlab{}.
\newblock \showarticletitle{Approximations of {{Stochastic Hybrid Systems}}}.
\newblock \bibinfo{journal}{\emph{IEEE Trans. Automat. Control}}
  \bibinfo{volume}{54}, \bibinfo{number}{6} (\bibinfo{year}{2009}),
  \bibinfo{pages}{1193--1203}.
\newblock


\bibitem[\protect\citeauthoryear{Khakpour and Mousavi}{Khakpour and
  Mousavi}{2015}]%
        {khakpour_NotionsConformanceTesting_2015}
\bibfield{author}{\bibinfo{person}{Narges Khakpour} {and}
  \bibinfo{person}{Mohammad~Reza Mousavi}.} \bibinfo{year}{2015}\natexlab{}.
\newblock \showarticletitle{Notions of {{Conformance Testing}} for
  {{Cyber}}-{{Physical Systems}}: {{Overview}} and {{Roadmap}}}. In
  \bibinfo{booktitle}{\emph{26th {{International Conference}} on {{Concurrency
  Theory}} ({{CONCUR}})}}, Vol.~\bibinfo{volume}{42}. \bibinfo{pages}{18--40}.
\newblock


\bibitem[\protect\citeauthoryear{Larsen and Legay}{Larsen and Legay}{2016}]%
        {larsen_StatisticalModelChecking_2016}
\bibfield{author}{\bibinfo{person}{Kim~G. Larsen} {and} \bibinfo{person}{Axel
  Legay}.} \bibinfo{year}{2016}\natexlab{}.
\newblock \showarticletitle{Statistical {{Model Checking}}: {{Past}},
  {{Present}}, and {{Future}}}. In \bibinfo{booktitle}{\emph{Leveraging
  {{Applications}} of {{Formal Methods}}, {{Verification}} and {{Validation}}:
  {{Foundational Techniques}}}}. \bibinfo{pages}{3--15}.
\newblock


\bibitem[\protect\citeauthoryear{Legay, Delahaye, and Bensalem}{Legay
  et~al\mbox{.}}{2010}]%
        {legay_StatisticalModelChecking_2010}
\bibfield{author}{\bibinfo{person}{Axel Legay}, \bibinfo{person}{Beno{\^i}t
  Delahaye}, {and} \bibinfo{person}{Saddek Bensalem}.}
  \bibinfo{year}{2010}\natexlab{}.
\newblock \showarticletitle{Statistical {{Model Checking}}: {{An Overview}}}.
\newblock In \bibinfo{booktitle}{\emph{Runtime {{Verification}}}}.
  Vol.~\bibinfo{volume}{6418}. \bibinfo{pages}{122--135}.
\newblock


\bibitem[\protect\citeauthoryear{Liu and Althoff}{Liu and Althoff}{2018}]%
        {liu_ReachsetConformanceForward_2018}
\bibfield{author}{\bibinfo{person}{Stephan~B. Liu} {and}
  \bibinfo{person}{Matthias Althoff}.} \bibinfo{year}{2018}\natexlab{}.
\newblock \showarticletitle{Reachset {{Conformance}} of {{Forward Dynamic
  Models}} for the {{Formal Analysis}} of {{Robots}}}. In
  \bibinfo{booktitle}{\emph{{{IEEE}}/{{RSJ International Conf.}} on
  {{Intelligent Robots}} and {{Systems}} ({{IROS}})}}.
  \bibinfo{pages}{370--376}.
\newblock


\bibitem[\protect\citeauthoryear{López, Núñez, and Rodríguez}{López
  et~al\mbox{.}}{2006}]%
        {lopez_SpecificationTestingImplementation_2006}
\bibfield{author}{\bibinfo{person}{Natalia López}, \bibinfo{person}{Manuel
  Núñez}, {and} \bibinfo{person}{Ismael Rodríguez}.}
  \bibinfo{year}{2006}\natexlab{}.
\newblock \showarticletitle{Specification, Testing and Implementation Relations
  for Symbolic-Probabilistic Systems}.
\newblock \bibinfo{journal}{\emph{Theoretical Computer Science}}
  \bibinfo{volume}{353}, \bibinfo{number}{1-3} (\bibinfo{year}{2006}),
  \bibinfo{pages}{228--248}.
\newblock


\bibitem[\protect\citeauthoryear{Majumdar, Saha, Ueda, and Yazarel}{Majumdar
  et~al\mbox{.}}{2013}]%
        {majumdar2013compositional}
\bibfield{author}{\bibinfo{person}{Rupak Majumdar}, \bibinfo{person}{Indranil
  Saha}, \bibinfo{person}{Koichi Ueda}, {and} \bibinfo{person}{Hakan Yazarel}.}
  \bibinfo{year}{2013}\natexlab{}.
\newblock \showarticletitle{Compositional equivalence checking for models and
  code of control systems}. In \bibinfo{booktitle}{\emph{52nd IEEE Conference
  on Decision and Control}}. \bibinfo{pages}{1564--1571}.
\newblock


\bibitem[\protect\citeauthoryear{Maler and Nickovic}{Maler and
  Nickovic}{2004}]%
        {maler2004monitoring}
\bibfield{author}{\bibinfo{person}{Oded Maler} {and} \bibinfo{person}{Dejan
  Nickovic}.} \bibinfo{year}{2004}\natexlab{}.
\newblock \showarticletitle{Monitoring temporal properties of continuous
  signals}.
\newblock In \bibinfo{booktitle}{\emph{Formal Techniques, Modelling and
  Analysis of Timed and Fault-Tolerant Systems}}.
  \bibinfo{publisher}{Springer}, \bibinfo{pages}{152--166}.
\newblock


\bibitem[\protect\citeauthoryear{{MathWorks, Inc.}}{{MathWorks, Inc.}}{2019a}]%
        {LKA-ref}
\bibfield{author}{\bibinfo{person}{{MathWorks, Inc.}}}
  \bibinfo{year}{2019}\natexlab{a}.
\newblock \bibinfo{title}{lane-keeping Assist System,}.
\newblock
  \bibinfo{howpublished}{\url{https://www.mathworks.com/help/mpc/ug/lane-keeping-assist-system-using-model-predictive-control.htmll}}.
\newblock
\newblock
\shownote{Accessed: 2019-7-15.}


\bibitem[\protect\citeauthoryear{{MathWorks, Inc.}}{{MathWorks, Inc.}}{2019b}]%
        {MPCToolbox}
\bibfield{author}{\bibinfo{person}{{MathWorks, Inc.}}}
  \bibinfo{year}{2019}\natexlab{b}.
\newblock \bibinfo{title}{Model Predictive Control Toolbox}.
\newblock
  \bibinfo{howpublished}{\url{https://www.mathworks.com/help/mpc/ug/lane-keeping-assist-system-using-model-predictive-control.html}}.
\newblock
\newblock
\shownote{Accessed: 2019-7-15.}


\bibitem[\protect\citeauthoryear{{MathWorks, Inc.}}{{MathWorks, Inc.}}{2019c}]%
        {PowerSim}
\bibfield{author}{\bibinfo{person}{{MathWorks, Inc.}}}
  \bibinfo{year}{2019}\natexlab{c}.
\newblock \bibinfo{title}{PowerSim}.
\newblock
  \bibinfo{howpublished}{\url{https://www.mathworks.com/help/physmod/sps/index.html?s_tid=CRUX_lftnav}}.
\newblock
\newblock
\shownote{Accessed: 2019-7-15.}


\bibitem[\protect\citeauthoryear{{MathWorks, Inc.}}{{MathWorks, Inc.}}{2019d}]%
        {simpowerPV}
\bibfield{author}{\bibinfo{person}{{MathWorks, Inc.}}}
  \bibinfo{year}{2019}\natexlab{d}.
\newblock \bibinfo{title}{SimPower}.
\newblock
  \bibinfo{howpublished}{\url{https://www.mathworks.com/help/physmod/sps/examples/250-kw-grid-connected-pv-array.html}}.
\newblock
\newblock
\shownote{Accessed: 2019-7-15.}


\bibitem[\protect\citeauthoryear{Pan, Cheng, Saigol, Lee, Yan, Theodorou, and
  Boots}{Pan et~al\mbox{.}}{2017}]%
        {pan2017agile}
\bibfield{author}{\bibinfo{person}{Yunpeng Pan}, \bibinfo{person}{Ching-An
  Cheng}, \bibinfo{person}{Kamil Saigol}, \bibinfo{person}{Keuntaek Lee},
  \bibinfo{person}{Xinyan Yan}, \bibinfo{person}{Evangelos Theodorou}, {and}
  \bibinfo{person}{Byron Boots}.} \bibinfo{year}{2017}\natexlab{}.
\newblock \showarticletitle{Agile autonomous driving using end-to-end deep
  imitation learning}.
\newblock \bibinfo{journal}{\emph{arXiv preprint arXiv:1709.07174}}
  (\bibinfo{year}{2017}).
\newblock


\bibitem[\protect\citeauthoryear{Peacock}{Peacock}{1983}]%
        {peacock_TwodimensionalGoodnessoffitTesting_1983}
\bibfield{author}{\bibinfo{person}{John~A. Peacock}.}
  \bibinfo{year}{1983}\natexlab{}.
\newblock \showarticletitle{Two-dimensional goodness-of-fit testing in
  astronomy}.
\newblock \bibinfo{journal}{\emph{Monthly Notices of the Royal Astronomical
  Society}} \bibinfo{volume}{202}, \bibinfo{number}{3} (\bibinfo{year}{1983}),
  \bibinfo{pages}{615--627}.
\newblock


\bibitem[\protect\citeauthoryear{Pinisetty, Schneider, and Sands}{Pinisetty
  et~al\mbox{.}}{2018}]%
        {pinisetty_RuntimeVerificationHyperproperties_2018}
\bibfield{author}{\bibinfo{person}{Srinivas Pinisetty},
  \bibinfo{person}{Gerardo Schneider}, {and} \bibinfo{person}{David Sands}.}
  \bibinfo{year}{2018}\natexlab{}.
\newblock \showarticletitle{Runtime Verification of Hyperproperties for
  Deterministic Programs}. In \bibinfo{booktitle}{\emph{Proceedings of the 6th
  {{Conference}} on {{Formal Methods}} in {{Software Engineering}} -
  {{FormaliSE}} '18}}. \bibinfo{pages}{20--29}.
\newblock


\bibitem[\protect\citeauthoryear{Roehm, Oehlerking, Woehrle, and Althoff}{Roehm
  et~al\mbox{.}}{2016}]%
        {roehm_ReachsetConformanceTesting_2016}
\bibfield{author}{\bibinfo{person}{Hendrik Roehm}, \bibinfo{person}{Jens
  Oehlerking}, \bibinfo{person}{Matthias Woehrle}, {and}
  \bibinfo{person}{Matthias Althoff}.} \bibinfo{year}{2016}\natexlab{}.
\newblock \showarticletitle{Reachset {{Conformance Testing}} of {{Hybrid
  Automata}}}. In \bibinfo{booktitle}{\emph{19th {{International Conference}}
  on {{Hybrid Systems}}: {{Computation}} and {{Control}} {{(HSCC)}}}}.
  \bibinfo{pages}{277--286}.
\newblock


\bibitem[\protect\citeauthoryear{Roohi, Wang, West, Dullerud, and
  Viswanathan}{Roohi et~al\mbox{.}}{2017}]%
        {roohi_StatisticalVerificationToyota_2017}
\bibfield{author}{\bibinfo{person}{Nima Roohi}, \bibinfo{person}{Yu Wang},
  \bibinfo{person}{Matthew West}, \bibinfo{person}{Geir~E. Dullerud}, {and}
  \bibinfo{person}{Mahesh Viswanathan}.} \bibinfo{year}{2017}\natexlab{}.
\newblock \showarticletitle{Statistical {{Verification}} of the {{Toyota
  Powertrain Control Verification Benchmark}}}. In
  \bibinfo{booktitle}{\emph{Proceedings of the 20th {{International
  Conference}} on {{Hybrid Systems}}: {{Computation}} and {{Control}} (HSCC)}}.
  \bibinfo{pages}{65--70}.
\newblock


\bibitem[\protect\citeauthoryear{Ryabtsev and Strichman}{Ryabtsev and
  Strichman}{2009}]%
        {ryabtsev2009translation}
\bibfield{author}{\bibinfo{person}{Michael Ryabtsev} {and}
  \bibinfo{person}{Ofer Strichman}.} \bibinfo{year}{2009}\natexlab{}.
\newblock \showarticletitle{Translation validation: From simulink to c}. In
  \bibinfo{booktitle}{\emph{International Conference on Computer Aided
  Verification}}. \bibinfo{pages}{696--701}.
\newblock


\bibitem[\protect\citeauthoryear{Sproston}{Sproston}{2000}]%
        {sproston_DecidableModelChecking_2000}
\bibfield{author}{\bibinfo{person}{Jeremy Sproston}.}
  \bibinfo{year}{2000}\natexlab{}.
\newblock \showarticletitle{Decidable {{Model Checking}} of {{Probabilistic
  Hybrid Automata}}}. In \bibinfo{booktitle}{\emph{Formal {{Techniques}} in
  {{Real}}-{{Time}} and {{Fault}}-{{Tolerant Systems}}}}.
  \bibinfo{pages}{31--45}.
\newblock


\bibitem[\protect\citeauthoryear{Trivedi, Bobbio, and Muppala}{Trivedi
  et~al\mbox{.}}{2017}]%
        {trivedi_ReliabilityAvailabilityEngineering_2017}
\bibfield{author}{\bibinfo{person}{Kishor~S. Trivedi}, \bibinfo{person}{Andrea
  Bobbio}, {and} \bibinfo{person}{Jogesh Muppala}.}
  \bibinfo{year}{2017}\natexlab{}.
\newblock \bibinfo{booktitle}{\emph{Reliability and {{Availability
  Engineering}}: {{Modeling}}, {{Analysis}}, and {{Applications}}}}.
\newblock


\bibitem[\protect\citeauthoryear{van~der Vaart and Wellner}{van~der Vaart and
  Wellner}{1996}]%
        {van1996glivenko}
\bibfield{author}{\bibinfo{person}{A.~W. van~der Vaart} {and}
  \bibinfo{person}{Jon~A Wellner}.} \bibinfo{year}{1996}\natexlab{}.
\newblock \showarticletitle{Glivenko-Cantelli Theorems}.
\newblock In \bibinfo{booktitle}{\emph{Weak Convergence and Empirical
  Processes}}. \bibinfo{publisher}{Springer}, \bibinfo{pages}{122--126}.
\newblock


\bibitem[\protect\citeauthoryear{Wang, Nalluri, Bonakdarpour, and Pajic}{Wang
  et~al\mbox{.}}{2021}]%
        {wang2021csf}
\bibfield{author}{\bibinfo{person}{Yu Wang}, \bibinfo{person}{Siddhartha
  Nalluri}, \bibinfo{person}{Borzoo Bonakdarpour}, {and}
  \bibinfo{person}{Miroslav Pajic}.} \bibinfo{year}{2021}\natexlab{}.
\newblock \showarticletitle{Statistical model checking for hyperproperties}. In
  \bibinfo{booktitle}{\emph{IEEE Computer Security Foundations Symposium
  (CSF)}}. \bibinfo{pages}{To appear}.
\newblock


\bibitem[\protect\citeauthoryear{Wang, Roohi, West, Viswanathan, and
  Dullerud}{Wang et~al\mbox{.}}{2015}]%
        {wang_StatisticalVerificationDynamical_2015}
\bibfield{author}{\bibinfo{person}{Yu Wang}, \bibinfo{person}{Nima Roohi},
  \bibinfo{person}{Matthew West}, \bibinfo{person}{Mahesh Viswanathan}, {and}
  \bibinfo{person}{Geir~E. Dullerud}.} \bibinfo{year}{2015}\natexlab{}.
\newblock \showarticletitle{Statistical {{Verification}} of {{Dynamical Systems
  Using Set Oriented Methods}}}. In \bibinfo{booktitle}{\emph{18th {{Int.
  Conf.}} on {{Hybrid Systems}}: {{Computation}} and {{Control}} (HSCC)}}.
  \bibinfo{pages}{169--178}.
\newblock


\bibitem[\protect\citeauthoryear{Wang, Zarei, Bonakdarpour, and Pajic}{Wang
  et~al\mbox{.}}{2019}]%
        {wang_StatisticalVerificationHyperproperties_2019}
\bibfield{author}{\bibinfo{person}{Yu Wang}, \bibinfo{person}{Mojtaba Zarei},
  \bibinfo{person}{Borzoo Bonakdarpour}, {and} \bibinfo{person}{Miroslav
  Pajic}.} \bibinfo{year}{2019}\natexlab{}.
\newblock \showarticletitle{Statistical Verification of Hyperproperties for
  Cyber-Physical Systems}.
\newblock \bibinfo{journal}{\emph{ACM Transactions on Embedded Computing
  Systems}} \bibinfo{volume}{18}, \bibinfo{number}{5s} (\bibinfo{year}{2019}),
  \bibinfo{pages}{1--23}.
\newblock


\bibitem[\protect\citeauthoryear{Zarei, Wang, and Pajic}{Zarei
  et~al\mbox{.}}{2020}]%
        {zarei2020statistical}
\bibfield{author}{\bibinfo{person}{Mojtaba Zarei}, \bibinfo{person}{Yu Wang},
  {and} \bibinfo{person}{Miroslav Pajic}.} \bibinfo{year}{2020}\natexlab{}.
\newblock \showarticletitle{Statistical verification of learning-based
  cyber-physical systems}. In \bibinfo{booktitle}{\emph{Proceedings of the 23rd
  International Conference on Hybrid Systems: Computation and Control (HSCC)}}.
  \bibinfo{pages}{1--7}.
\newblock


\bibitem[\protect\citeauthoryear{Zhang, She, Ratschan, Hermanns, and
  Hahn}{Zhang et~al\mbox{.}}{2010}]%
        {zhang_SafetyVerificationProbabilistic_2010}
\bibfield{author}{\bibinfo{person}{Lijun Zhang}, \bibinfo{person}{Zhikun She},
  \bibinfo{person}{Stefan Ratschan}, \bibinfo{person}{Holger Hermanns}, {and}
  \bibinfo{person}{Ernst~Moritz Hahn}.} \bibinfo{year}{2010}\natexlab{}.
\newblock \showarticletitle{Safety {{Verification}} for {{Probabilistic Hybrid
  Systems}}}. In \bibinfo{booktitle}{\emph{Computer {{Aided Verification}}}}.
  \bibinfo{pages}{196--211}.
\newblock


\end{thebibliography}
